\documentclass{CSML}

\def\dOi{12(2:3)2016}
\lmcsheading%
{\dOi}
{1--28}
{}
{}
{Jan.~\phantom05, 2015}
{Apr.~25, 2016}
{}

\ACMCCS{[{\bf Theory of computation}]: Formal languages and automata
  theory---Automata over infinite objects; Logic---Logic and verification}
\keywords{data automata, counter automata, two-variable logic}

\usepackage[utf8]{inputenc}

\usepackage{amssymb} 
\usepackage{amsmath} 
\usepackage{amsthm}
\usepackage{mathbbol} 
\usepackage{stmaryrd}  

\usepackage{hyperref}
\usepackage{xspace}

\usepackage{pgf}
\usepackage{tikz}
\usetikzlibrary{trees}
\usetikzlibrary{shapes.geometric}
\usetikzlibrary{matrix}

\newtheorem{theorem}{Theorem}
\newtheorem{lemma}[theorem]{Lemma}
\newtheorem{proposition}[theorem]{Proposition}

\newtheorem{remark}{Remark}

\newcommand\remove[1]{}

\usepackage{ifthen}

\newcommand\ra{\ensuremath{\rightarrow}}

\newcommand\vass{\textup{VASS}\xspace}
\newcommand\bvass{\textup{BVASS}\xspace}
\newcommand\ebvass{\textup{EBVASS}\xspace}

\newcommand\dataeq{\ensuremath{\sim}}
\newcommand\neighbor{\ensuremath{\mathord{+}1}}
\newcommand\descendant{\ensuremath{\mathord{<}}}
\newcommand\fo{\mathsf{FO}\xspace}   
\newcommand\mso{\mathsf{MSO}\xspace} 
\newcommand\emso{\mathsf{E}\mso}     
\newcommand\fotwo{\ensuremath{\fo^2(\descendant,\neighbor,\mathord{\sim})}\xspace}

\newcommand\emsotwo{\ensuremath{\emso^2(\descendant,\neighbor,\mathord{\sim})}\xspace}
\newcommand\xpath{\textup{XPath}\xspace}

\newcommand\dad{\textup{DTA}$^\#$\xspace}
\newcommand\wdad{\textup{DTA}\xspace}

\newcommand\A{\ensuremath{\mathbb{A}}\xspace}
\newcommand\B{\mathbb{B}}
\newcommand\D{\mathbb{D}}
\newcommand\E{\mathbb{E}}
\newcommand\F{\mathbb{F}}
\newcommand\N{\ensuremath{\mathbb{N}}\xspace}

\newcommand\Aa{\mathcal{A}}
\newcommand\Ba{\mathcal{B}}
\newcommand\Ca{\mathcal{C}}
\newcommand\Da{\mathcal{D}}
\newcommand\Ea{\mathcal{E}}

\newcommand\merc[3]{\mathsf{C}_{#1}\mathpunct{\circleddash}\mathsf{C}_{#2}\to\mathsf{C}_{#3}}
\newcommand\meru[3]{\ifthenelse{\equal{#1}{}}{}{\mathsf{C}_{#1}}\mathpunct{\circleddash}\ifthenelse{\equal{#2}{}}{}{\mathsf{C}_{#2}}\to_1\mathsf{C}_{#3}}

\newcommand\Trees{\mathsf{Trees}}
\newcommand\Forests{\mathsf{Forests}}

\newcommand\tree{\boldsymbol{t}}
\newcommand\atree{\boldsymbol{a}}
\newcommand\btree{\boldsymbol{b}}
\newcommand\dtree{\boldsymbol{d}}

\newcommand\set[1]{\ensuremath{\{#1\}}}

\newcommand\stree{\boldsymbol{s}}
\newcommand\wtree{\boldsymbol{w}}

\newcommand\leaf{\mathsf{leaf}}

\newcommand\last{\mathsf{last}}

\newcommand\fcns{\ensuremath{\mathit{fcns}}\xspace}

\newcommand\rootstate{$\Ba$-state\xspace}

\newcommand\EE{\mathsf{E}}
\newcommand\NS{\EE_{\rightarrow}}

\newcommand\PC{\EE_{\downarrow}}

\newcommand\SAD{\EE_{\downdownarrows}}
\newcommand\SFS{\EE_{\rightrightarrows}}

\renewcommand\parallel{\mathbin{\!<\mkern-12mu>\!}}

\def\frew#1#2#3#4#5#6#7#8{
\setbox0=\hbox{$#6 #7 #1 #8$}%
\setbox1=\hbox{$#6 #7 #2 #8$}%
\ifdim \wd0>\wd1 \rlap{\rlap{\hbox to \wd0{#5}}%
                            {\hbox to\wd0{\hfil\lower #3\box1\relax\hfil}}}{\raise #4\box0}%
\else \rlap{\rlap{\hbox to \wd1{#5}}{\hbox to\wd1{\hfil\raise #4\box0\relax\hfil}}}{\lower #3\box1}%
\fi
}

\newcommand{\toto}{\longleftrightarrow}

\title[\fotwo on data trees]{\fotwo on data trees, data tree automata and branching vector addition systems}
\author[F.~Jacquemard]{Florent Jacquemard}
\author[L.~Segoufin]{Luc Segoufin}
\author{J\'er\'emie Dimino}
\address{INRIA and ENS Cachan}
\email{\{florent.jacquemard, luc.segoufin\}@inria.fr}

\begin{document}

\begin{abstract}
  A data tree is an unranked ordered tree where each node carries a label from
  a finite alphabet and a datum from some infinite domain.  We consider the two
  variable first order logic \fotwo over data trees.  Here $+1$ refers to the
  child and the next sibling relations while $<$ refers to the descendant and
  following sibling relations.  Moreover,~$\sim$ is a binary predicate testing
  data equality.  We exhibit an automata model, denoted \dad, that is more
  expressive than \fotwo but such that emptiness of \dad and satisfiability of
  \fotwo are inter-reducible.  This is proved via a model of counter tree
  automata, denoted \ebvass, that extends Branching Vector Addition Systems
  with States (\bvass) with extra features for merging counters. We show that,
  as decision problems, reachability for \ebvass, satisfiability of \fotwo and
  emptiness of \dad are equivalent.
\end{abstract}

\maketitle

\section*{Introduction}

A data tree is an unranked ordered tree where each node carries a label from a
finite alphabet and a datum from some infinite domain. Together with the
special case of data words, they have been considered in the realm of program
verification, as they are suitable to model the behavior of concurrent,
communicating or timed systems, where data can represent \textit{e.g.}, process
identifiers or time stamps~\cite{Alur12AAA,BCGK12fossacs,Bouyer02ipl}.  Data
trees are also a convenient model for XML documents~\cite{BojanczykMSS09jacm},
where data represent attribute values or text contents.  Therefore finding
decidable logics for this model is a central problem as it has applications in
most reasoning tasks in databases and in verification.

Several logical formalisms and models of automata over data trees have been
proposed. Many of them were introduced in relation to \xpath, the standard
formalism to express properties of XML documents. Although satisfiability of
\xpath in the presence of data values is undecidable, automata models were
introduced for showing decidability of several 
data-aware fragments~\cite{FS11,BojanczykMSS09jacm,Fig10,F09,JL08}.

As advocated in~\cite{BojanczykMSS09jacm}, the logic 
\fotwo can be seen as a relevant fragment of \xpath. 
Here \fotwo refers to the two-variable fragment of first
order logic over unranked ordered data trees, with predicates for the child and the next
sibling relations ($+1$), predicates for the descendant and following sibling
relations ($<$) and a predicate for testing data equality between two nodes ($\sim$).
Over data words, \fotwo was shown to be decidable by a reduction to Petri Nets or,
equivalently, Vector Addition Systems with States (\vass)~\cite{BDMSS11}.  
It is also shown in~\cite{BojanczykMSS09jacm} that reachability for
Branching Vector Addition Systems with States, \bvass, 
reduces to satisfiability of \fotwo over data trees.  
The model of \bvass,  extends \vass with a natural branching feature for running
on trees, see~\cite{acl10} for a survey of the various formalisms equivalent to
\bvass.  As the reachability of \bvass is a long standing open problem, showing
decidability of finite satisfiability for \fotwo seems unlikely in the near future.

This paper is a continuation of the work of~\cite{BDMSS11,BojanczykMSS09jacm}.
We introduce a model of counter automata, denoted \ebvass, and show that
satisfiability of \fotwo is inter-reducible to reachability in \ebvass. 
This model extends \bvass by allowing new features for merging counters.  
In a \bvass the value of a counter at a node $x$ in a binary tree
is the sum of the values of that counter at the children of $x$, 
plus or minus some constant specified by the transition relation. 
In \ebvass constraints can be added modifying this behavior. In particular
(see Section \ref{sec-dad-counter} for a more precise definition) it can enforce the
following at node $x$: 
one of the counters of its left child and one
of the counters of its right child are decreased by the same arbitrary number $n$, 
then the sum is performed as for \bvass, and finally, 
one of the resulting counters is increased by $n$.

The reduction from \fotwo to \ebvass goes via a new
model of data tree automata, denoted \dad. 
Our first result (Section~\ref{sec-fotwo-dad}) shows that languages of data trees 
definable in \fotwo are also recognizable by \dad.  
Moreover the construction of the automaton from the formula is effective. 
Our automata model is a non-trivial extension
from data words to data trees of the Data Automata (DA) model
of~\cite{BDMSS11}, chosen with care in order to be powerful enough to capture the logic 
but also not too powerful in order to match its computational power. 
The obvious extensions of DA to data trees are either too weak to capture \fotwo 
or too expressive and undecidable (see Proposition~\ref{prop-undecid}).
Here we consider the strongest of these extensions, called $\wdad$, 
which is undecidable,
and restrict it into a model called \dad
with an associated emptiness problem is equivalent to satisfiability of \fotwo.

Our second result (Section~\ref{sec-dad-counter}) shows that the emptiness
problem for \dad reduces to the reachability problem for \ebvass.  Finally we
show in Section~\ref{sec-counter-fotwo} that the latter problem can be reduced
to the satisfiability of \fotwo, closing the loop.  Altogether, this implies
that showing (un)decidability of any of these problems would show
(un)decidability of the three of them.
Although this question of (un)decidability remains open,
the equivalence shown in this paper between the decidability of these three problems,
the definition of the intermediate model \dad and the techniques used for proving
the interreductions provides
a better understanding of the three problems, 
and in particular of the emptiness of the branching vector addition systems with states.

{\bf Related work.}
There are many other works introducing automata or logical formalism for data
words or data trees. Some of them are shown to be decidable using counter automata,
see for instance~\cite{DL-tocl08,JL08}. The link between counter automata and
data automata is not surprising as the latter
only compare data values via equality. 
Hence they are invariant under permutation of the data domain and therefore, 
often, it is enough
to count the number of data values satisfying some properties instead of
knowing their precise values.

\section{Preliminaries}\label{sec:prelim}

In this paper $\A$ or $\B$ denote finite alphabets while $\D$ denotes an infinite
data domain. We use $\E$ or $\F$ when we do not care whether the alphabet is
finite or not. We denote by $\E_\#$ the extension of an alphabet $\E$ with a new symbol $\#$
that does not occur in $\E$.

\paragraph{Unranked ordered data forests.}
We work with finite unranked ordered trees and forests over an alphabet~$\E$,
defined inductively as follows: for any $a \in \E$, $a$ is a tree. If
$\tree_1,\cdots,\tree_k$ is a finite non-empty sequence of trees then 
$\tree_1 + \cdots + \tree_k$ is a forest. If $\stree$ is a forest and $a \in \E$, then
$a(\stree)$ is a tree. The set of trees and forests over $\E$ are respectively
denoted $\Trees(\E)$ and $\Forests(\E)$.  A tree is called unary
(resp. binary) when every node has at most one (resp. two) children.
We use standard terminology for trees and forests defining nodes, roots,
leaves, parents, children, ancestors, descendants, following 
and preceding siblings.

Given a forest $\tree \in \Forests(\E)$, and a node $x$ of $\tree$, we denote
by $\tree(x)$ the label of $x$ in $\tree$.

We say that two forests $\tree_1 \in \Forests(\E_1)$ 
and $\tree_2 \in \Forests(\E_2)$ 
\emph{have the same domain} if there is a bijection from the
nodes of $\tree_1$ to the nodes of $\tree_2$ 
that respects the parent and the next-sibling relations. 
In this case we identify the nodes of $\tree_1$ with
the nodes of $\tree_2$ and the difference between $\tree_1$ and $\tree_2$ lies
only in the label associated to each node. 
Given two forests $\tree_1 \in \Forests(\E_1)$, 
$\tree_2 \in \Forests(\E_2)$ having the same domain, 
we define $\tree_1 \otimes \tree_2 \in \Forests(\E_1\times\E_2)$ 
as the forest over the same domain 
and such that for all nodes $x$, 
$\tree_1 \otimes \tree_2 (x) = \langle \tree_1(x),\tree_2(x)\rangle$.

The set of \emph{data forests} over a finite alphabet $\A$ 
and an infinite data domain $\D$  is defined as $\Forests(\A {\times} \D)$. 
Note that every 
$\tree \in \Forests(\A {\times} \D)$ 
can be decomposed into 
$\atree \in \Forests(\A)$ 
and $\dtree \in \Forests(\D)$ such that $\tree = \atree \otimes \dtree$.

\paragraph{Logics on data forests.}
A data forest of $\Forests(\A {\times} \D)$ can be seen as a relational model for
first order logic.  The domain of the model is the set of nodes in the
forest.  There is a unary relation $a(x)$ for all $a \in \A$ containing
the nodes of label $a$. There is a binary relation $x \dataeq y$ containing all
pairs of nodes carrying the same data value of $\D$, 
and binary relations $\NS(x,y)$ ($y$ is the sibling immediately next to $x$),
$\PC(x,y)$ ($x$ is the parent of $y$), and $\SFS$, $\SAD$ which are the non
reflexive transitive closures respectively of $\NS$ and $\PC$, minus
respectively $\NS$ and $\PC$ (\textit{i.e.}, they define two or more navigation steps).
The reason for this non-standard definition of $\SFS$ and $\SAD$ is that it
will be convenient that equality, $\NS$, $\PC$, $\SFS$ and $\SAD$ are disjoint
binary relations. We will often make use of the macro, $x \parallel y$, and say
that $x$ and $y$ are \emph{incomparable}, when none of $x = y$, $\NS(x,y)$,
$\PC(x,y)$, $\SFS(x,y)$ and $\SAD(x,y)$ holds.

Let \fotwo be the set of first order sentences with two variables built on top
of the above predicates. Typical examples of properties definable in \fotwo are
key constraints (all nodes of label $a$ have different data values), $\forall x
\forall y ~a(x) \land a(y) \land x \sim y \ra x=y$, and downward inclusion
constraints (every node $x$ of label $a$ has a node $y$ of label $b$ in its
subtree with the same data value), 
$\forall x \exists y ~a(x) \ra \bigl(b(y) \land x \sim y \land (\SAD(x,y) \lor \PC(x,y))\bigr)$.

We also consider the extension \emsotwo of \fotwo with existentially quantified
monadic second order variables.  Every formula of \emsotwo has the form
$\exists R_1\ldots \exists R_n\; \phi$ where $\phi$ is a \fotwo formula called
the \emph{core}, involving the variables $R_1, \ldots, R_n$ as unary
predicates.  The extension to full monadic second order logic is denoted
$\mso(<,+1,\dataeq)$.

We write $\mso(<,+1)$ for the set of formulas not using the $\sim$ predicates.
These formulas are ignoring the data values, \textit{i.e.}, 
they are classical monadic second-order formulas over forests.


\paragraph{Automata models for forests.} 
We will informally refer to automata and transducers 
for forests and unranked trees over a finite alphabet. 
The particular choice of a model of automata is not relevant here 
and we refer to~\cite[Chapters~1,6,8]{tata} for a detailed description. 
A set of forests accepted by an automaton is called a \emph{regular language} 
and regular languages are exactly those definable in $\mso(<,+1)$.

\paragraph{Automata models for data forests.}  
Given a data forest $\tree = \btree \otimes \dtree \in \Forests(\B {\times} \D)$ 
and a data value $d \in \D$, the \emph{class forest}\label{def:class-forest} 
$\tree[d]$ of $\tree$ associated to the datum $d$ 
is the forest of $\Forests(\B_\#)$ having the same domain as $\tree$
and such that $\tree[d](x)=\btree(x)$ if $\dtree(x)=d$ and $\tree[d](x)=\#$ otherwise.

\begin{figure}[h!]
\hspace{-8.7 cm}\begin{tikzpicture}
  [level distance=7mm,
   level 1/.style={sibling distance=28mm},
   level 2/.style={sibling distance=9mm}]
\node {$(a,1)$}
 child { node {$(b,1)$}
         child { node {$(c,1)$} }
         child { node {$(b,1)$} }
         child { node {$(a,1)$} }
       }
 child { node {$(b,1)$}
         child { node {$(a,2)$} }
         child { node {$(b,2)$} }
         child { node {$(a,1)$} }
       };
\hspace{6cm}
\node {$a$}  [level distance=7mm,
   level 1/.style={sibling distance=15mm},
   level 2/.style={sibling distance=5mm}]
 child { node {$b$}
         child { node {$c$} }
         child { node {$b$} }
         child { node {$a$} }
       }
 child { node {$b$}
         child { node {$\#$} }
         child { node {$\#$} }
         child { node {$a$} }
       };
\hspace{4cm}
\node {$\#$}[level distance=7mm,
   level 1/.style={sibling distance=15mm},
   level 2/.style={sibling distance=5mm}]
 child { node {$\#$}
         child { node {$\#$} }
         child { node {$\#$} }
         child { node {$\#$} }
       }
 child { node {$\#$}
         child { node {$a$} }
         child { node {$b$} }
         child { node {$\#$} }
       };
\end{tikzpicture}
  \caption{A forest $\tree$ followed by its class forests
    $\tree[1]$ and $\tree[2]$}
  \label{fig-class-forest}
\end{figure}
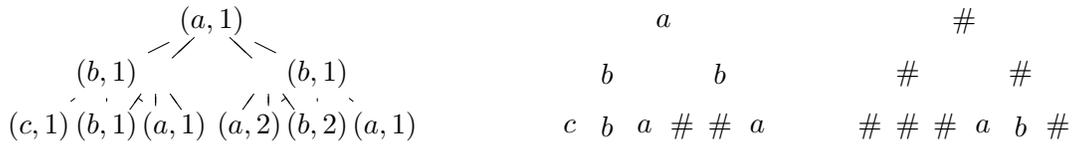

\noindent We now define two models of automata over data trees. The first and most general
one is a straightforward generalization to forests of the automata model over
data words of~\cite{BDMSS11}. The second one adds a restriction in order to
avoid undecidability.

\paragraph{General Data Forest Automata model: \wdad.}  
A \wdad is a pair $(\Aa,\Ba)$ where $\Aa$ is a non-deterministic
letter-to-letter transducer taking as input a forest in $\Forests(\A)$ and
returning a forest in $\Forests(\B)$ with the same domain, while $\Ba$ is a
forest automaton taking as input a forest in $\Forests(\B_\#)$. Intuitively a
\wdad works as follows on a forest $\tree=\atree\otimes\dtree$: first the
transducer $\Aa$ relabels the nodes of $\atree$ into $\btree$ and the forest
automaton $\Ba$ has to accept all class forests of~$\btree\otimes\dtree$. 

More formally a data forest $\tree=\atree\otimes\dtree \in \Forests(\A\times\D)$ 
is accepted by $(\Aa,\Ba)$ iff
\begin{enumerate}
\item there exists $\btree \in \Forests(\B)$ such that
$\btree$ is a possible output of $\Aa$ on $\atree$ and,
\item for all $d\in\D$, 
      the class forest $(\btree\otimes\dtree)[d] \in \Forests(\B_\#)$ is accepted by $\Ba$.
\end{enumerate}

\noindent
Over data words this model was shown to be decidable~\cite{BDMSS11}. 
Unfortunately it is undecidable over data trees.
\begin{proposition}\label{prop-undecid}
Emptiness of \wdad is undecidable.
\end{proposition}
\begin{proof}
We show that \wdad can simulate the Class Automata
of~\cite{BojanczykLasota10lics}.  This latter model has an undecidable
emptiness problem, already when restricted to data words, 
\textit{i.e.}, forests of the form $\langle a_{1}, d_1\rangle + \ldots + \langle a_{m}, d_m\rangle$.  
It captures indeed the class of languages of words - without data - recognized by counter automata.
Like Data Automata, Class Automata are defined as pairs made of one transducer~$\Aa$ and one word automaton~$\Ba$.
However, 
the $\Ba$ part in the Class Automata model has access to the label of the nodes that are not in the
class, while it sees only $\#$ in the Data Automata case. 
This extra power implies undecidability.

\newcommand\classproj{\llbracket d \rrbracket}
\newcommand\classtreeprime{\tree'\classproj}
\newcommand\classword{\wtree\classproj}

We assume two finite alphabets $\A$ and $\B$, writing the latter
\textit{in extenso} as $\B = \{ b_1,\ldots, b_n \}$.  
A \emph{class automaton} over
$\A \times \D$ is a pair $\Ca = (\Aa, \Ba)$ where $\Aa$ is a non-deterministic
letter-to-letter word transducer from $\A$ into $\B$ and $\Ba$ is a word
automaton taking as input words over the alphabet $\B \times \{ 0,1 \}$.
In order to define the acceptance of data words by class automata, 
we shall use a notion of class word associated to a data word $\wtree = \btree\otimes\dtree$
and a value $d \in \D$, denoted $\classword$, defined as 
the word having the same domain as $\wtree$
and such that, for every node $x$ of $\wtree$, 
$\classword(x) = \langle \btree(x), 1\rangle$ if $\dtree(x) = d$
and 
$\classword(x) = \langle \btree(x), 0\rangle$ otherwise.
A data word $\wtree = \atree\otimes\dtree$ is accepted by $\Ca$ iff
\begin{enumerate} 
\item there exists a word $\btree$ over $\B$ such that
$\btree$ is a possible output of $\Aa$ on $\atree$ and,
\item for all $d\in\D$, the class word $(\btree\otimes \dtree)\classproj$ 
is accepted by $\Ba$.
\end{enumerate}

Given a class automaton $\Ca= (\Aa, \Ba)$ over $\A\times\D$, we construct a \wdad $\Ca'$ such that $\Ca$
accepts a data word iff $\Ca'$ accepts a data tree.
The idea of the reduction is that we replace each letter $b_i$ by a tree of
depth $i$. Hence, even if $b_i$ is replaced by $\#$ during the run of $\Ca'$ (conversion to class word),
this label can still be recovered.

Let $\mathbb{O}$ be a new alphabet containing the two symbols $b$ and $\#$. 
For any symbol $s$ and $1 \leq i \leq n$, let $s^i$ be the unary data tree of depth $i$
defined recursively by: $s^1 = s$ and $s^{i+1} = s(s^i)$.
We associate to a data word $\wtree = \langle b_{i_1}, d_1\rangle + \ldots +
\langle b_{i_m}, d_m\rangle$ a data forest 
$\hat{\wtree} \in \Forests(\mathbb{O} \times \D)$ defined by 
$\hat{\wtree} = \bigl(\langle b,d_1\rangle^{i_1+1} + \ldots + \langle b, d_m\rangle^{i_m+1}\bigr)$.

From the word automaton $\Ba$ we can construct a forest automaton $\Ba'$ accepting exactly the set of
class forests $\hat{\wtree}[d]$ such that $\wtree\classproj$ is accepted by
$\Ba$, for all $d\in\D$. 

The best way to see this is to use $\mso(<,+1)$ logic. The language recognized
by $\Ba$ can be defined by a formula $\varphi$ of $\mso(<,+1)$. The formula
corresponding to $\Ba'$ is constructed by replacing in $\varphi$ each atom of
the form $\langle b_i,1\rangle$ by $b_i(x)$ and each atom of the form $\langle
b_i,0\rangle$ by a formula testing that $x$ has label $\#$ and that the
subtree rooted at $x$ has depth $i$.

\medskip

From there it is now easy to construct an $\Aa'$ such that the \wdad $(\Aa',\Ba')$
accepts a data forest iff the class automaton $\Ca=(\Aa,\Ba)$ accepts a data word.
\end{proof}

\paragraph{Restricted Data Forest Automata model: \dad.}  
The second data tree automata model we consider is defined as \wdad with a
restriction on $\Ba$. The restriction makes sure that $\Ba$ ignores repeated
and contiguous occurrences of $\#$ symbols.  This ensures that for each class
forest $\tree[d]$, not only the automata cannot see the label of a node whose
data value is not $d$, but also can not see the shape of subtrees of nodes
whose data value differs from $d$. In particular it can no longer count the number
of $\#$ symbols in a subtree and the undecidability proof of
Proposition~\ref{prop-undecid} no longer works.

A set $L \subseteq \Forests(\B)$ is called $\#$-\emph{stuttering} 
iff it is closed  
under the rules depicted in Figure~\ref{fig-rules}. 
Intuitively these rules should be understood as follows: 
if a subforest is matched by a left-hand side of a rule
(when the variables $x$ and $y$ 
 are replaced by (possibly empty) forests), 
then replacing this subforest by the corresponding right-hand side 
(with the same variable replacement)  
yields a forest also in $L$,
and the other way round. 

\begin{figure}[h!]
\footnotesize
\[\hspace{-8 pt}
\begin{array}{rclcrclcrclrclc}
\begin{minipage}[t]{10mm}
\vspace{-3.5mm}
\hspace{4mm}\begin{tikzpicture}
  [level distance=7mm,
   level 1/.style={sibling distance=10mm},
   level 2/.style={sibling distance=5mm}]
\node {$\#$}
 child { node {$\#$}
         child { node {$x$} }
       };
\end{tikzpicture}
\end{minipage}
&
\vspace{-0mm}\toto
&
\begin{minipage}[t]{5mm}
\vspace{-3.5mm}
\begin{tikzpicture}
  [level distance=7mm,
   level 1/.style={sibling distance=10mm},
   level 2/.style={sibling distance=5mm}]
\node {$\#$}
 child { node {$x$} };
\end{tikzpicture}
\end{minipage}
& ~~ &
\begin{minipage}[t]{5mm}
\vspace{-3.5mm}
\begin{tikzpicture}
  [level distance=7mm,
   level 1/.style={sibling distance=10mm},
   level 2/.style={sibling distance=5mm}]
\node {$\#$}
 child { node {$x$} };
\end{tikzpicture}
\end{minipage}
+ \; \#
&
\vspace{-0mm}\toto
&
\begin{minipage}[t]{5mm}
\vspace{-3.5mm}
\begin{tikzpicture}
  [level distance=7mm,
   level 1/.style={sibling distance=10mm},
   level 2/.style={sibling distance=5mm}]
\node {$\#$}
 child { node {$x$} };
\end{tikzpicture}
\end{minipage}
& \quad &
\# \; + 
\begin{minipage}[t]{5mm}
\vspace{-3.5mm}
\begin{tikzpicture}
  [level distance=7mm,
   level 1/.style={sibling distance=10mm},
   level 2/.style={sibling distance=5mm}]
\node {$\#$}
 child { node {$x$} };
\end{tikzpicture}
\end{minipage}
&
\vspace{-0mm}\toto
&
\begin{minipage}[t]{5mm}
\vspace{-3.5mm}
\begin{tikzpicture}
  [level distance=7mm,
   level 1/.style={sibling distance=10mm},
   level 2/.style={sibling distance=5mm}]
\node {$\#$}
 child { node {$x$} };
\end{tikzpicture}
\end{minipage}
& \quad &
\begin{minipage}[t]{5mm}
\vspace{-3.5mm}
\begin{tikzpicture}
  [level distance=7mm,
   level 1/.style={sibling distance=10mm},
   level 2/.style={sibling distance=5mm}]
\node {$\#$}
 child { node {$x$} };
\end{tikzpicture}
\end{minipage}
+ y
&
\vspace{-0mm}\toto
&
y +
\begin{minipage}[t]{5mm}
\vspace{-3.5mm}
\begin{tikzpicture}
  [level distance=7mm,
   level 1/.style={sibling distance=10mm},
   level 2/.style={sibling distance=5mm}]
\node {$\#$}
 child { node {$x$} };
\end{tikzpicture}
\end{minipage}
\end{array}
\]
\caption{Closure rules for $\#$-stuttering sets. $x$ represents an arbitrary forest.}
\label{fig-rules}
\end{figure}
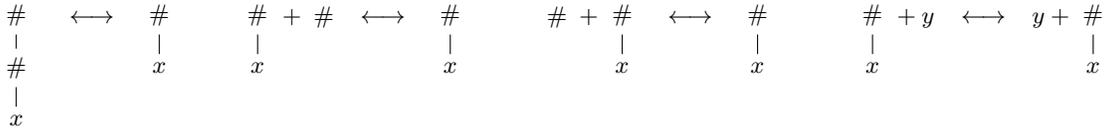

For instance if $L$ is $\#$-stuttering and contains the trees $\tree[1]$ and
$\tree[2]$ of Figure~\ref{fig-class-forest}, then it should also contain the
trees in Figure~\ref{fig-class-closure}.

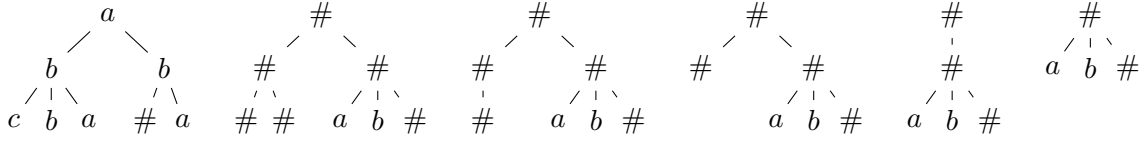
\begin{figure}
\[
\begin{tikzpicture}[baseline]
\node {$a$}  [level distance=7mm,
   level 1/.style={sibling distance=15mm},
   level 2/.style={sibling distance=5mm}]
 child { node {$b$}
         child { node {$c$} }
         child { node {$b$} }
         child { node {$a$} }
       }
 child { node {$b$}
         child { node {$\#$} }
         child { node {$a$} }
       };
\end{tikzpicture}
\quad\!
\begin{tikzpicture}[baseline]
\node {$\#$}[level distance=7mm,
   level 1/.style={sibling distance=15mm},
   level 2/.style={sibling distance=5mm}]
child { node {$\#$}
        child { node {$\#$} }
        child { node {$\#$} }
      }
child { node {$\#$}
        child { node {$a$} }
        child { node {$b$} }
        child { node {$\#$} }
      };
\end{tikzpicture}
\quad\!
\begin{tikzpicture}[baseline]
\node {$\#$}[level distance=7mm,
   level 1/.style={sibling distance=15mm},
   level 2/.style={sibling distance=5mm}]
child { node {$\#$}
        child { node {$\#$} }
      }
child { node {$\#$}
        child { node {$a$} }
        child { node {$b$} }
        child { node {$\#$} }
      };
\end{tikzpicture}
\quad\!
\begin{tikzpicture}[baseline]
\node {$\#$}[level distance=7mm,
   level 1/.style={sibling distance=15mm},
   level 2/.style={sibling distance=5mm}]
child { node {$\#$} }
child { node {$\#$}
        child { node {$a$} }
        child { node {$b$} }
        child { node {$\#$} }
      };
\end{tikzpicture}
\quad\!
\begin{tikzpicture}[baseline]
\node {$\#$}[level distance=7mm,
   level 1/.style={sibling distance=15mm},
   level 2/.style={sibling distance=5mm}]
child { node {$\#$}
        child { node {$a$} }
        child { node {$b$} }
        child { node {$\#$} }
      };
\end{tikzpicture}
\quad\!
\begin{tikzpicture}[baseline]
\node {$\#$}[level distance=7mm,
   level 1/.style={sibling distance=5mm},
   level 2/.style={sibling distance=5mm}]
         child { node {$a$} }
         child { node {$b$} }
         child { node {$\#$} };
\end{tikzpicture}
\]
\caption{Closure of $\{ \tree[1], \tree[2] \}$ of Figure~\ref{fig-class-forest}.}
\label{fig-class-closure}
\end{figure}

\medskip Typical examples of $\#$-stuttering languages are those testing that
no two nodes of label $a$ occur in $\tree[d]$ (key constraint) or that each
node of label $a$ has a descendant of label $b$ in $\tree[d]$ (inclusion
constraint).  Other typical $\#$-stuttering languages are those defined by
formulas of the form $\forall x \forall y\; a(x) \wedge b(y) \wedge x \dataeq y
\to \neg \SFS(x, y)$. Indeed the  $\#$-stuttering rules do not affect the
relationship between pairs of nodes not labeled by~$\#$.

Typical examples of languages that are not $\#$-stuttering are
those counting the number of nodes of label~$\#$.  
Note that $\#$-stuttering
languages are closed under union and intersection.

\medskip

We define \dad as those \wdad $(\Aa,\Ba)$ such that the language recognized by $\Ba$ is $\#$-stuttering.
\medskip

We conclude this section with the following simple lemma whose proof is a
straightforward Cartesian product construction. We use the term \emph{letter
  projection} for a relabeling function defined as $h : \A' \to \A$, where $\A$
and $\A'$ are alphabets.

\begin{lemma} \label{lem:conjunction-disjunction}
The class of \dad languages is closed 
under union, intersection and letter projection.
\end{lemma}

\section{From \texorpdfstring{\fotwo}{FO2} to \texorpdfstring{\dad}{DTA}}
\label{sec-fotwo-dad}

In this section we show the following result.
\begin{theorem}\label{th-fotwodad}
Given a formula $\phi$ in \fotwo, there exists a \dad,
effectively computable from $\phi$,
accepting exactly the set of data forests satisfying~$\phi$.
\end{theorem}

The proof works in two steps. 
In the first step we provide a normal form for sentences of \fotwo 
that is essentially an \emsotwo formula 
whose core is a conjunction of simple formula of \fotwo. 
In a second step, 
we show that each of the conjunct can be translated into a \dad, 
and we conclude using
composition of these automata by intersection, 
see Lemma~\ref{lem:conjunction-disjunction}.

\subsection{Intermediate Normal Form} \label{app:normal-form}
We show first that every \fotwo formula $\phi$ can be transformed into an equivalent \emsotwo formula
in \emph{intermediate normal form}:
\begin{equation*}
\exists R_1\cdots \exists R_k \bigwedge_i \chi_i 
\end{equation*}
where each $\chi_i$ has one of the following forms:
\begin{eqnarray}
 & \forall x \forall y & \alpha(x) \wedge \beta(y) \wedge \delta(x, y) \to \gamma(x, y) 
\label{eq:inf-fafa}\\
 & \forall x \exists y & \alpha(x) \to (\beta(y) \wedge \delta(x, y) \wedge
 \epsilon(x, y)) 
\label{eq:inf-faex}
\end{eqnarray}
where each of $\alpha$ and $\beta$ is a \emph{type}, 
that is, a conjunction of unary predicates or their negation 
(these unary predicates are either from \A or from $R_1,\ldots,R_k$, 
\textit{i.e.}, introduced by the existentially quantified variables),
$\delta(x, y)$ is either $x \dataeq y$ or $x \not\dataeq y$,
$\gamma(x, y)$ is one of 
      $\neg \SFS(x, y)$, $\neg \SAD(x, y)$ 
      or $\neg (x \parallel y)$,
and
$\epsilon(x, y)$ is one of
$x = y$, 
$\NS(x, y)$, 
$\NS(y, x)$, 
$\PC(x, y)$,
$\PC(y, x)$,
$\SFS(x, y)$, 
$\SFS(y, x)$, 
$\SAD(x, y)$,
$\SAD(y, x)$,
$x \parallel y$ or $\mathit{false}$.

This normal form is obtained by simple syntactical manipulation 
similar to the one given in~\cite{BDMSS11} for the data words case, and detailed below. 

\paragraph{\bf Scott normal form.}
We first transform the formula $\phi$ into Scott Normal Form
obtaining an \emsotwo formula of the form:
\[
\psi = 
\exists R_1 \ldots \exists R_m\; \forall x \forall y\; \chi \wedge 
\bigwedge_{i} \forall x \exists y\; \chi_i
\]
where $\chi$ and every $\chi_i$ are quantifier free, 
and $R_1, \ldots R_m$ are new unary predicates (called monadic).
This transformation is standard: 
a new unary predicate $R_\theta$ is introduced
for each subformula $\theta(x)$ with one free variable for marking the
nodes where the subformula holds. The subformula $\theta(x)$ is then
replaced by $R_\theta(x)$ and a conjunct 
$\forall x \bigl(R_\theta(x) \leftrightarrow \theta(x)\bigr)$ is added. 
This yields a formula in the desired normal form.

\paragraph{\bf From Scott to intermediate normal form.}
We show next that every conjunct of the core of the formula $\psi$ in Scott Normal Form
can be replaced by an equivalent conjunction of formulas of the
form~\eqref{eq:inf-fafa} or~\eqref{eq:inf-faex}, possibly by adding 
new quantifications with unary predicates upfront.

\paragraph{\bf Case $\forall x \forall y\; \chi$.}
Recall that with our definition, the binary relations $\NS$, $\SFS$,
$\PC$, $\SAD$, $\parallel$ and $=$ are pairwise disjoint.
Hence we can rewrite $\forall x \forall y\; \chi$ into an equivalent \fotwo formula 
in the following form,
\[
\forall x \forall y \left(\!%
\begin{array}{rrcl}
       & x = y & \!\to\! & \psi_{=}(x,y)\\
\!\wedge\! & x \parallel y & \!\to\! & \psi_{\parallel}(x,y) \bigr)
\end{array}
\begin{array}{rrcl}
\!\wedge\! & \NS(x, y) & \!\to\! & \psi_{\!\to\!}(x,y)\\
\!\wedge\! & \PC(x, y) & \!\to\! & \psi_{\downarrow}(x,y)
\end{array}
\begin{array}{rrcl}
\!\wedge\! & \SFS(x, y) & \!\to\! & \psi_{\rightrightarrows}(x,y)\\
\!\wedge\! & \SAD(x, y) & \!\to\! & \psi_{\downdownarrows}(x,y)
\end{array}\!%
\right)
\]
where every subformula $\psi_*$ is quantifier free and only involves the predicate
$\dataeq$ together with unary predicates. 
They can be obtained from $\chi$ 
via conjunctive normal form and De Morgan's law.
The resulting formula is equivalent to the conjunction
\[
\begin{array}{cll}
       & \forall x \exists y & \bigl( x = y \wedge \psi_{=}(x,y)\bigl)\\
\wedge & \forall x \exists y & 
  \bigl( \neg\last(x) \to (\NS(x, y) \wedge \psi_{\to}(x,y)) \bigr)\\
\wedge & \forall x \exists y & 
  \bigl( \neg\leaf(x) \to (\PC(x, y) \wedge \psi_{\downarrow}(x,y)) \bigr)\\
\wedge & \forall x \forall y & \SFS(x, y) \to \psi_{\rightrightarrows}(x,y)\\
\wedge & \forall x \forall y & \SAD(x, y) \to \psi_{\downdownarrows}(x,y)\\                  
\wedge & \forall x \forall y & x \parallel y \to \psi_{\parallel}(x,y)\\                  
\end{array}
\]
where $\leaf(x)$ is a new predicate denoting the leaves of the forest and
$\last(x)$ is also a new predicate denoting nodes having no right sibling.  The
predicate $\leaf$ is specified by the following formulas, that have the desired form. 
\[
\begin{array}{ll}
\forall x \forall y & \PC(x,y) \land \leaf(x) \to \mathit{false}\\
\forall x \exists y & \lnot\leaf(x) \to \PC(x, y)
\end{array}
\]
Similar formulas specify the predicate $\last$.

\noindent
The first three conjuncts, with quantifier prefix $\forall x \exists y$, will be
treated later when dealing with the second case.

\noindent
For the next three conjuncts, putting $\neg\psi_{\rightrightarrows}$, $\neg \psi_{\downdownarrows}$, 
$\neg\psi_{\parallel}$ in disjunctive normal form 
(with an exponential blowup), 
we rewrite $\psi_{\rightrightarrows}$, $\psi_{\downdownarrows}$,
$\psi_{\parallel}$ as a conjunction of formulas of the form $\neg(\alpha(x)
\wedge \beta(y) \wedge \delta(x,y))$, 
where 
$\alpha$, $\beta$, and $\delta$ are as in (\ref{eq:inf-fafa}).
%
By distribution of conjunction over implication, and by contraposition, 
we obtain for the 3 cases an equivalent conjunction of formulas of the
following form (matching the desired form~\eqref{eq:inf-fafa}) 
\[
\begin{array}{ll}
\forall x \forall y & 
 \alpha(x) \wedge \beta(y) \wedge \delta(x,y) \to \neg \SFS(x, y)\\
\forall x \forall y & 
 \alpha(x) \wedge \beta(y) \wedge \delta(x,y) \to \neg \SAD(x, y)\\
\forall x \forall y & 
 \alpha(x) \wedge \beta(y) \wedge \delta(x,y) \to \neg (x \parallel y)\\
\end{array}
\]

\paragraph{\bf Case $\forall x \exists y\; \chi$.}
We first transform $\chi$ (with an exponential blowup) 
into an equivalent disjunction of the form
\[
\chi' = \bigvee_j \alpha_j(x) \wedge \beta_j(y) \wedge \delta_j(x,y) \wedge \epsilon_j(x,y)
\]
where $\alpha_j$, $\beta_j$, $\delta_j$ and $\epsilon_j$
are as in (\ref{eq:inf-faex}).
Next, in order to eliminate the disjunctions, we add a new
monadic second-order variables $R_{\chi, j}$, that we existentially quantify
upfront of the global formula, and transform $\forall x \exists
y\; \chi'$ into the conjunction
\[
\bigwedge_j
\forall x \exists y\; (\alpha_j(x) \wedge R_{\chi, j}(x)) \to (\beta_j(y) \wedge \delta_j(x,y) \wedge \epsilon_j(x,y))
\wedge
\forall x \exists y\; \bigl(\bigvee_j R_{\chi, j}(x)\bigr)
\]
The first conjuncts express that 
if $R_{\chi, j}(x)$ holds, then there exists a node $y$ 
such that the corresponding conjunct of $\chi'$ holds, 
and the last conjunct expresses that 
for all node $x$, at least one of the $R_{\chi, j}(x)$ must hold and can be
rewritten as $\forall x \exists y\; \bigl(\bigwedge \lnot R_{\chi, j}(x) \to
\mathit{false}\bigr)$. Now all the conjuncts are as in \eqref{eq:inf-faex} and
we are done.


\subsection{Case analysis for constructing \texorpdfstring{\dad}{DTA} from intermediate normal forms} 
We now show how to transform a formula in intermediate normal form into a \dad. 
Let \A be the initial alphabet and let $\A'$ be the new alphabet formed
by combining letters of \A with the newly quantified unary predicates $R_1,\ldots, R_k$.  
By closure of \dad under intersection and letter projection (Lemma~\ref{lem:conjunction-disjunction}), 
it is enough to construct a \dad automaton for each simple formula of the form~\eqref{eq:inf-fafa}
or~\eqref{eq:inf-faex}, accepting the data forests in $\Forests(\A'\times \D)$ satisfying the formula.

We do a case analysis depending on the atoms involved in the formula of the
form~\eqref{eq:inf-fafa} or~\eqref{eq:inf-faex}. 
For each case we construct a \dad $(\Aa,\Ba)$ 
recognizing the set of data forests satisfying the formula.  
The construction borrows several ideas from the data word case~\cite{BDMSS11}, but
some extra work is needed as the tree structure is more complicated.
In the discussion below, a node whose label satisfies the type $\alpha$ will
be called an \emph{$\alpha$-node}.
Many of the cases build on generic constructions that we described 
in the following remark.
\begin{remark}\label{remark-color-class}
  A \dad $(\Aa,\Ba)$ can be used to distinguish one 
  specific data value, by recoloring, with $\Aa$, all the nodes carrying the
  data value, and checking, with $\Ba$, the correctness of the recoloring.  We
  will then say that $(\Aa,\Ba)$ \emph{marks a data value using the new color
    $c$}.  This can be done as follows.  The transducer $\Aa$ marks
  (i.e. relabel the node by adding to its current label an extra color) a node $x$ with
  this data value with a specific new color $c'$.  At the same time it guesses
  all the nodes sharing the same data value as $x$ and marks each of them with
  a new color $c$.  Then, the forest automaton $\Ba$ checks, for every data
  value, that either none of the nodes are marked with $c$ or $c'$, or that all
  nodes not labeled with $\#$ are marked with $c$ or $c'$ and that $c'$ occurs
  exactly once in the same class forest.  Note that this defines a
  $\#$-stuttering language.  It is now clear that for the run to be accepting,
  $\Aa$ must color exactly one data value and that all the nodes carrying this
  data value must be marked with $c$ or $c'$.  The transducer $\Aa$ can then
  build on this fact for checking other properties.
\end{remark}

\noindent
A generic example of the usefulness of this remark is given below.
Once an arbitrary data value is marked with a color $c$, 
then a property of the form 
$\forall x \forall y\; \alpha(x) \wedge \beta(y) \wedge x \not\dataeq y \to \gamma(x, y)$ 
is a conjunction of 
$\forall x \forall y\; \alpha(x) \land c(x) \wedge \beta(y) \wedge \lnot c(y) \to \gamma(x,y)$ 
with 
$\forall x \forall y\; \alpha(x) \land \lnot c(x) \wedge \beta(y) \wedge x \not\dataeq y 
 \to \gamma(x, y)$. 
 The first part, 
$\forall x \forall y\; \alpha(x) \land c(x) \wedge \beta(y) \wedge \lnot c(y) 
 \to \gamma(x, y)$ 
is now a regular property and can therefore be tested by $\Aa$. 
Hence it is enough to consider the case where $x$ does not carry the marked data value. 
The same reasoning holds if two data values are marked or if the formula starts with a 
$\forall x \exists y$ quantification. 
We will use this fact implicitly in the case analysis below.

Given a data forest, a \emph{vertical path} is a set of nodes containing exactly
one leaf and all its ancestors and nothing else.  A \emph{horizontal path} is a
set of nodes containing one node together with all its siblings and nothing else.

\medskip\noindent
We start with formulas of the form~\eqref{eq:inf-fafa}. 

\paragraph{\bf Case 1:}
$\forall x \forall y\; \alpha(x) \wedge \beta(y) \wedge x \dataeq y \to \gamma(x, y)$,
where $\gamma(x, y)$ is as in (\ref{eq:inf-fafa}).
These formulas express a property of pairs of nodes with the same data value.
We have seen that those are $\#$-stuttering languages that can be tested
by the forest automaton $\Ba$ solely (\textit{i.e.}, by a \dad with $\Aa$ doing
nothing).

\paragraph{\bf Case 2:}
$\forall x \forall y\; \alpha(x) \wedge \beta(y) 
 \wedge x \not\dataeq y \to \neg \SFS(x, y)$.
This formula expresses that a data forest cannot contain 
an $\alpha$-node having a $\beta$-node with a different data value
as a sibling to its right, except if it is the next-sibling.
Let $X$ be an horizontal path in a data forest $\tree$ containing at least one
$\alpha$-node, and let $x$ be the leftmost $\alpha$-node in $X$. 
Let $d$ be the data value of $x$. 
Consider an $\alpha$-node $x'$ and a $\beta$-node $y'$ that make the formula
false within $X$, in particular we have $\SFS(x',y')$. Then, if $y'$ has a data
value different from $d$ we already have $\SFS(x,y')$ and the formula is
also false for the pair $(x,y')$.
Hence the validity of the formula within $X$ can be tested over pairs $(x',y')$ such that
either $x'$ or $y'$ has data value $d$.

\noindent
With this discussion in mind we construct $(\Aa,\Ba)$ as follows. In every horizontal path $X$
containing one $\alpha$-node, the transducer $\Aa$ identify the leftmost
occurrence $x$ of an $\alpha$-node in $X$, and marks it with a new color $c'$, 
and marks all the nodes of $X$ with the same data value as $x$ with a color $c$.  
As in Remark~\ref{remark-color-class}, the forest automaton $\Ba$ checks that the
guesses are correct, \textit{i.e.} it accepts only forests in which every horizontal
path $X$ satisfy one of the following conditions: 
$X$ contains one occurrence of the color $c'$ 
and all other nodes of $X$ not labeled with $\#$ are marked with $c$, 
or $X$ contains none of the colors $c$ and $c'$ at all. 
All these properties define regular and $\#$-stuttering languages, 
and hence can be checked by a forest automaton $\Ba$.

Assuming this, the transducer $\Aa$ rejects if there are some unmarked
$\beta$-nodes occurring as a right sibling (except for the next-sibling) of a
marked $\alpha$-node or there is an unmarked $\alpha$-node as
left sibling, except for the previous sibling, of a marked $\beta$-node.  
As explained in Remark~\ref{remark-color-class}, this is a regular property.

\paragraph{\bf Case 3:}
$\forall x \forall y\; \alpha(x) \wedge \beta(y) \wedge x \not\dataeq y \to \neg \SAD(x, y)$.
The property expressed by this formula is similar to the previous case,
replacing the right sibling relationship with the descendant relationship.

Let $X$ be a vertical path in a data forest $\tree$ containing at least one
$\alpha$-node, and let $x$ be the $\alpha$-node in $X$ the closest to
the root. Let $d$ be the data value of $x$.
Consider an $\alpha$-node $x'$ and a $\beta$-node $y'$ that make the formula
false within $X$, in particular we have $\SAD(x',y')$. Then, if $y'$ has a data
value different from $d$ we already have $\SAD(x,y')$ and the formula is
also false for the pair $(x,y')$.
Hence the validity of the formula within $X$ can be tested over pairs $(x',y')$ such that
either $x'$ or $y'$ has data value $d$.

\noindent
The construction of $(\Aa,\Ba)$ is similar to the previous case, except that different
vertical paths may share some nodes.
The transducer $\Aa$ marks all the $\alpha$-nodes 
that have no $\alpha$-node as ancestor, with a new color $c'$. 
Then, for every node $x$ marked $c'$, $\Aa$ guesses all the
nodes inside the subtree rooted at $x$ having the same data value as
$x$ and mark them with a new color $c$.
As in Remark~\ref{remark-color-class}, the forest automaton $\Ba$ checks that the guesses of
colors are correct for each vertical path (see also the previous case).

Assuming this, the transducer $\Aa$ rejects if there are an unmarked
$\beta$-node that is a descendant, but not a child, of a
marked $\alpha$-node or there is an unmarked $\alpha$-node as
an ancestor, except for the parent, of a marked $\beta$-node.  
This is a regular property that can be checked by $\Aa$ in conjunction with the marking,
following the principles of Remark~\ref{remark-color-class}. 

\paragraph{\bf Case 4:}
$\forall x \forall y\; \alpha(x) \wedge \beta(y) \wedge x \not\dataeq y \to 
\neg (x \parallel y)$.
The formula expresses that every two nodes of type respectively $\alpha$ and
$\beta$ and with different data values cannot be incomparable.
Recall that two nodes are incomparable if they are not ancestors and not siblings.

\paragraph{\bf Subcase 4.1:} There exist two $\alpha$-nodes that are incomparable. 

Let $x_1$ and $x_2$ be two incomparable $\alpha$-nodes and let $z$ 
be their least common ancestor
(see Figure~\ref{fig:case41a}).
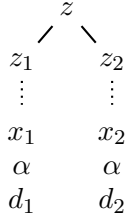
\begin{figure}
\[
\begin{tikzpicture} [-,thick]
\tikzstyle{level 1}=[level distance=7mm,sibling distance=12mm]
\tikzstyle{level 2}=[level distance=14mm,sibling distance=5mm]
\node {$z$}
  child { node {$z_1$} 
	  child[dotted] { node {$\begin{array}{c} x_1\\ \alpha\\ d_1 \end{array}$} } } 
  child { node {$z_2$} 
	  child[dotted] { node {$\begin{array}{c} x_2\\ \alpha\\ d_2 \end{array}$} } }; 
\end{tikzpicture}
\]
\caption{Subcase 4.1 in the proof of Theorem~\ref{th-fotwodad}.}
\label{fig:case41a}
\end{figure}

We can choose $x_1$ and $x_2$ such that none of the
$\alpha$-nodes are incomparable with $z$ or sibling of~$z$, because if this was not
the case then there is an $\alpha$-node $x_3$ incomparable with $z$ or sibling of~$z$, 
and therefore $x_3$ is incomparable with $x_1$, and we can replace $x_2$ with
$x_3$, continuing with their least common ancestor, 
a node which is strictly higher than $z$. 
Let $z_1$ and $z_2$ be the children of $z$ 
that are respectively ancestors of $x_1$ and $x_2$. 
Note that by construction, $z_1\neq z_2$. 
If $x_1=z_1$ and there is an $\alpha$-node $x_3$ in the subtree of $z$, 
different from $x_1$ and incomparable with $z_2$, 
then we replace $x_1$ by $x_3$ and proceed. 
In other words we ensure that if $x_1 = z_1$ 
then there is no $\alpha$-node incomparable with $z_2$ in the subtree of $z$. 
We proceed similarly
to enforce that if $x_2=z_2$ then there is no
$\alpha$-node incomparable with $z_1$ in the subtree of $z$. 
Notice that we cannot have at the same time $x_1 = z_1$ and $x_2 = z_2$
because we assumed $x_1$ and $x_2$ to be incomparable. 
All these properties can be specified in $\mso(<,+1)$ and therefore can be tested 
by a forest automaton.  
Let $d_1$ and $d_2$ be the respective data values of $x_1$ and $x_2$ (possibly $d_1=d_2$).

Consider now a $\beta$-node $y$ whose data value is neither $d_1$ nor $d_2$.
If $y$ is incomparable with $z$ or sibling of $z$, then the formula cannot be true as
it is contradicted by $(x_1,y)$.  If $y$ is an ancestor of $z$ then, as no
$\alpha$-node is incomparable with $z$, none is incomparable with $y$. 
Hence the formula can only be true with such $y$.
Assume now that $y$ is inside the subtree of~$z$. 
If $y = z_1$ and $x_2 \neq z_2$, then the formula is contradicted by $(x_2, y)$.
If $y = z_1$ and $x_2 = z_2$, then, by hypothesis, 
there is no $\alpha$-node incomparable with $y$ in the subtree of $z$, 
and there is no $\alpha$-node incomparable with $y$ outside the subtree of $z$,
hence altogether, the formula holds for $y$.
If $y \neq z_1$ and $y$ is a descendant of $z_1$, 
then the formula is contradicted by $(x_2, y)$.
The cases where $y$ is descendant of $z_2$ are symmetric:
in this case, the formula can only be true if $y = z_2$ and $x_1 = z_1$.
In the remaining cases $y$ is in the subtree of $z$ and not in the subtrees of
$z_1$ and $z_2$, making the formula false. Indeed, in each of these cases, either
$(x_1, y)$ or $(x_2, y)$ contradicts the formula.
To summarize, the only cases making the formula true
are when $y$ is an ancestor of $z$, 
or $y = z_1 \wedge x_2 = z_2$,
or $y = z_2 \wedge x_1 = z_1$.

With this discussion in mind, this case can be solved as follows:
The transducer $\Aa$ guesses the nodes of $x_1$, $x_2$, $z_1$, $z_2$ and $z$
and checks that they satisfy the appropriate properties. 
Moreover, $\Aa$ guesses whether $d_1=d_2$ 
and marks the data values of $x_1$ and $x_2$ accordingly
with one or two new colors. 
The forest automaton $\Ba$ will then check that the data values are marked
appropriately as in Remark~\ref{remark-color-class}.

Moreover $\Aa$ checks that for all marked $\beta$-nodes there is no
$\alpha$-node incomparable with it and with a different data value, a regular
property as explained in Remark~\ref{remark-color-class}.
It now remains for $\Aa$ to check that every unmarked $\beta$-node~$y$ behaves
according to the discussion above: 
$y$ is an ancestor of $z$ or $y = z_1$ and $x_2 = z_2$ or $y = z_2$ and $x_1 = z_1$.
This is a regular property testable by $\Aa$.

\paragraph{\bf Subcase 4.2:} There are no two incomparable $\alpha$-nodes. 

Let $x$ be an $\alpha$-node such that no $\alpha$-node is a descendant of $x$. 
By hypothesis, all $\alpha$-nodes are either ancestors or siblings of $x$. 
Let $d$ be the data value of $x$.
We distinguish between several subcases depending on whether 
there are other $\alpha$-nodes that are siblings of $x$ or not.

If there is an $\alpha$-node $x'$ that is a sibling of $x$, then let $d'$
be its data value (possibly $d=d'$).  Consider
now a $\beta$-node $y$ whose data value is neither $d$ nor $d'$.  
Then, in order to make the formula true, $y$ must be an ancestor or a sibling of $x$.

In this case, the transducer $\Aa$ guesses the nodes $x$ and $x'$ and marks the
corresponding data values with one or two new colors (according to whether $d = d'$
or not).  The forest automaton $\Ba$ will then check that the data values are marked
correctly as explained in Remark~\ref{remark-color-class}.
For the marked $\beta$-nodes, the property expressed by the formula
is regular and can also be checked by $\Aa$.
It remains for $\Aa$ to check that every unmarked $\beta$-nodes is 
either an ancestor of $x$ or a sibling of $x$.

Now, if there are no $\alpha$-nodes that are sibling of $x$,
and $y$ is a $\beta$-node whose data value is not $d$, 
then in order to make the formal true, $y$ cannot be incomparable with $x$, 
and therefore, $y$ can be an ancestor, a descendant or a sibling of $x$.

In this second case, the transducer $\Aa$ guesses the node $x$, 
marks its data value using a new color. 
The forest automaton $\Ba$ will then check that the data values were marked
correctly as explained in Remark~\ref{remark-color-class}.
The transducer $\Aa$ checks that all marked $\beta$-nodes make the formula true
(a regular property),
and that all unmarked $\beta$-nodes are not incomparable with $x$.

\medskip

We now turn to formulas of the form \eqref{eq:inf-faex}.

\paragraph{\bf Case 5:} $\forall x \exists y\; \alpha(x) \to (\beta(y) \wedge x \dataeq y
\wedge \epsilon(x, y))$, where $\epsilon(x, y)$ is as in \eqref{eq:inf-faex}.
These formulas express properties of nodes with the same data value. 
Moreover they express a regular property over all $\tree[d]$. 
Therefore can be treated by the forest automaton $\Ba$ as for the case 1.

\paragraph{\bf Case 6:}
$\forall x \exists y\; \alpha(x) \to (\beta(y) \wedge x \not\dataeq y \wedge \NS(x, y))$.
This formula expresses that every $\alpha$-node has a next sibling of type $\beta$
with a different data value.
The transducer $\Aa$ marks every $\alpha$-node $x$,
with a new color $c$ and checks that the next-sibling of $x$ is a $\beta$-node.
The forest automaton $\Ba$ accepts only the forests such that
for every node marked with $c$, its right sibling is labeled with~$\#$.

\paragraph{\bf Cases 7, 8, 9:}
The formulae of form 
$\forall x \exists y\; \alpha(x) \to (\beta(y) \wedge x \not\dataeq y 
 \wedge \epsilon(x, y))$
where $\epsilon(x, y)$ is one of $\NS(y, x)$, $\PC(x, y)$, $\PC(y, x)$ 
are treated similarly.

\paragraph{\bf Case 10:}
$\forall x \exists y\; \alpha(x) \to (\beta(y) \wedge x \not\dataeq y \wedge \SFS(x, y))$.
This formula expresses that every $\alpha$-node 
must have a $\beta$-node as a right sibling, but not as its next-sibling,
and with a different data value.

Let $X$ be an horizontal path. Let $y$ be the rightmost $\beta$-node of $X$ 
and $d$ be its data value. 
Consider now an $\alpha$-node $x$ of $X$ with a data value different from $d$. 
Then either $x$ is at the left of the previous-sibling of $y$, 
and $y$ can serve as the desired witness, 
or $x$ has no witness and the formula is false.

The transducer $\Aa$, for each horizontal path $X$ containing an
$\alpha$-node, marks its rightmost $\beta$-node $y$ with a new color $c'$,
guesses all the nodes of $X$ with the same data value as $y$ and marks them
with a new color $c$. 
Then it checks that every unmarked $\alpha$-node of $X$
occurs at the left of the previous-sibling of $y$. 
The forest automaton $\Ba$ checks that the guesses are correct 
as in Remark~\ref{remark-color-class}: 
for each horizontal paths, either all elements are marked with $c$ or $c'$, or none.

\paragraph{\bf Cases 11, 12, 13:}
The constructions for the formulae
$\forall x \exists y\; \alpha(x) \to (\beta(y) \wedge x \not\dataeq y 
 \wedge \epsilon(x, y))$
where $\epsilon(x, y)$ is one of $\SFS(y, x))$, $\SAD(x, y))$, and $\SAD(y, x))$ 
are similar.

\paragraph{\bf Case 14:}
$\forall x \exists y\; \alpha(x) \to (\beta(y) \wedge x \not\dataeq y \wedge x \parallel y)$.
This formula expresses that every $\alpha$-node 
must have a incomparable $\beta$-node with a different data value.

\paragraph{\bf Subcase 14.1:} 
There exist two $\beta$-nodes that are incomparable.

Let $y_1$ and $y_2$ be two incomparable $\beta$-nodes and let $z$ be their least
common ancestor. Using the same reasoning as in subcase 4.1, we can choose
$y_1$ and $y_2$ such that none of the $\beta$-nodes is incomparable with $z$ or a
sibling of $z$.  Let $z_1$ and $z_2$ be the children of $z$ that are the
ancestors of $y_1$ and $y_2$ respectively.  By construction, $z_1\neq z_2$.
Using the same trick as in subcase 4.1,
we can ensure that if $y_1 = z_1$ then there is no $\beta$-node
incomparable with $z_2$, and
if $y_2 = z_2$ then there is no $\beta$-node incomparable with $z_1$. 
Moreover, we cannot have at the same time $y_1=z_1$ and $y_2=z_2$. 
Recall that all these properties can be tested by a forest automaton. 
Let $d_1$ and $d_2$ be the respective data values of $y_1$ and $y_2$ (possibly $d_1=d_2$). 

Consider now an $\alpha$-node $x$ whose data value is neither $d_1$ nor $d_2$. 
If $x$ is incomparable with $z$ or a sibling of $z$, then $y_1$ is a witness for $x$. 
If $x$ is an ancestor of $z$ then by hypothesis there is no $\beta$-node incomparable 
with $x$ and hence the formula is false. 
Assume now that $x$ is in the subtree rooted at $z$. 
If $x = z_1$ and $y_2 \neq z_2$, then $y_2$ is a $\beta$-node
incomparable with $x$ with a different data value, hence a witness for $x$ in the formula.
If $x = z_1$ and $y_2 = z_2$, 
then by hypothesis, there is no $\beta$-node incomparable with $x$
in the subtree of $z$, and since there are neither $\beta$-nodes incomparable with $x$
outside the subtree of $z$, the formula must be false.
If $x \neq z_1$ and $x$ is a descendant of $z_1$, 
then $y_2$ is a witness for $x$.
The cases where $x$ is a descendant of $z_2$ are symmetric.
In the remaining cases, $x$ is in the subtree of $z$ 
and not a descendant of $z_1$ or $z_2$.
In each of these cases, either $y_1$ or $y_2$ is a witness for $x$.

With this discussion in mind, this case can be solved as follows:
The transducer $\Aa$ guesses the nodes of $y_1$, $y_2$, $z_1$, $z_2$ and $z$ 
and checks that they satisfy the appropriate properties. 
Moreover, $\Aa$ guesses whether $d_1=d_2$ and marks accordingly the data values 
of $z_1$ and $z_2$ with one or two new colors. 
The forest automaton $\Ba$ will then check that the data values are marked appropriately, 
as in Remark~\ref{remark-color-class}.
Moreover $\Aa$ checks that for every marked $\alpha$-node, 
there exists a $\beta$-node making the formula true. 
It remains for $\Aa$ to check the three following properties: 
no unmarked $\alpha$-node occurs above $z$,
if $y_1=z_1$ then $z_2$ is not an unmarked $\alpha$-node,
and if $y_2=z_2$ then $z_1$ is not an unmarked $\alpha$-node.

\paragraph{\bf Subcase 14.2:} There are no two $\beta$-nodes that are incomparable. 

Let $y$ be an $\beta$-node such that no $\beta$-node is
a descendant of $y$. By hypothesis, all $\beta$-nodes are either ancestors
or siblings of $y$. Let $d$ be the data value of $y$.
We distinguish between several subcases depending on whether there are $\beta$-nodes
that are siblings of $y$ or not.

If there exists a $\beta$-node $y'$ that is a sibling of $y$, 
let $d'$ be its data value (possibly $d = d'$).
Consider an $\alpha$-node $x$ whose data value is neither $d$ nor $d'$.
If $x$ is incomparable with $y$, then $y$ is a witness for $x$. 
If $x$ is an ancestor or a sibling of $y$, 
then the formula cannot be true, because by hypothesis
every $\beta$-node cannot be incomparable with $x$.
If $x$ is 
a descendant of $y$, then $y'$ makes the formula true for that $x$.

Consider now the case where there are no $\beta$-node that are sibling of $y$.
Note that $y$ can have $\beta$-nodes among its ancestors.
Let $x$ be a $\alpha$-node that has data value different from $d$. 
If $x$ is not incomparable with $y$ then the formula must be false. 
Otherwise, $y$ is a witness for $x$.

The transducer $\Aa$ guesses the $\beta$-node $y$ and marks its data value using a new color. 
Then it checks whether there is an $\beta$-node $y'$ that is a sibling of $y$.
If yes, it guesses whether the value at $y'$ is the same as the value at $y$ or not,
and marks the data value of $y'$ using a new color.
The forest automaton $\Ba$ will then check that the data values are marked appropriately.
For marked $\alpha$-nodes, $\Aa$ checks the regular property making the formula true.
It now remains for $\Aa$ to check, in both cases, 
that every unmarked $\alpha$-node $x$
satisfy the appropriate condition described above,
\textit{i.e.},
 that $x$ is incomparable with $y$ or a descendant of $y$ 
 if there exists a sibling $y'$ 
 and that $x$ is incomparable with $y$ otherwise.

\paragraph{\bf Case 15:}
$\forall x \exists y\; \alpha(x) \to \mathit{false}$. 
It is sufficient to test with $\Aa$ that no $\alpha$-node is present in the forest.

\section{From \texorpdfstring{\dad}{DTA} to \texorpdfstring{\ebvass}{EBVASS}}\label{sec-dad-counter}

In this section we show that the emptiness problem of \dad can be reduced to
the reachability of a counter tree automata model that extends \bvass, denoted~\ebvass. 
An \ebvass is a tree automaton equipped with counters.  It runs on
binary trees over a finite alphabet.  It can increase or decrease its
counters but cannot perform a zero test.  For \bvass, when going up in the
tree, the new value of each counter is the sum of its values at the left and
right child. An \ebvass can change this behavior using simple arithmetical
constraints.

The general idea of the reduction is as follows. 
\label{sketch:sec-dad-counter}
Let $(\Aa,\Ba)$ be a \dad. 
We want to construct
an automaton that recognizes exactly the projections of the data forests
accepted by $(\Aa,\Ba)$. Because this automaton does not have access to the data
values, the main difficulty is to simulate the runs of $\Ba$ on all class
forests. We will use counters for this purpose. The automaton will maintain the
following invariant: At any node~$x$ of the forest, for each state~$q$ of
$\Ba$, we have a counter that stores the number of data values~$d$ such that
$\Ba$ is in state $q$ at $x$ when running on the class forest associated to
$d$. In order to maintain this invariant we make sure that the  automata
model has the appropriate features for manipulating the counters. In
particular, moving up in the tree, in order to simulate $\Ba$, 
the automaton has to decide which data value occurring 
in the left subtree also appears in the right subtree.  
At the current node, each data value is associated to a state of $\Ba$ 
and, by the invariant property, a counter. In order to maintain the invariant for data
values occurring in both subtrees, for each pair $q,q'$ of states of $\Ba$, the
automaton guesses a number~$n$ (the number of data values being at the same
time in state $q$ in the left subtree and in state $q'$ in the right subtree),
removes $n$ from both associated counters and adds $n$ to the counter
corresponding to the state resulting from the application of the transition function of
$\Ba$ on $(q,q')$. This preserves the invariant property and a BVASS cannot do it, so we explicitly add this feature to our model. Once we
have done this, the counters from both sides are added like a BVASS would do. 
The $\#$-stuttering property of the language of $\Ba$ will ensure that this last operation
is consistent with the behavior of $\Ba$.
This is essentially what we do. But of course there are some nasty details. In
particular \dad run over unranked forests while \ebvass run over binary trees.

We start by defining the counter tree automata model and then we present
the reduction.

\subsection{Definition of \ebvass} 

An \ebvass is a tuple $(Q,\A,q_0,k,\delta,\chi)$ 
where \A is a finite alphabet, $Q$ is a finite set of states, $q_0 \in Q$ is the initial
state, $k\in\N$ is the number of counters, $\chi$ is a finite set of constraints of
the form $\merc{i_1}{i_2}{i} $with $1\leq i_1,i_2,i\leq k$, and $\delta$ is the set of transitions which are
of two kinds: $\epsilon$-transitions (subset denoted $\delta_\epsilon$) and
up-transitions (subset denoted $\delta_\text{u}$).

An $\epsilon$-transition is an element of $(Q\times\A) \times (Q\times U)$
where $U=\set{I_i,D_i~:~1\leq i \leq k}$ is the set of possible counter
updates: $D_i$ stands for \emph{decrement counter $i$} and $I_i$ stands for
\emph{increment counter $i$}. We view each element of $U$ as a vector over
$\set{-1,0,1}^k$ with only one non-zero position.
An up-transition is an element of $(Q\times\A\times Q) \times Q$.

Informally, an $\epsilon$-transition may change the current state and increment
or decrement one of the counters. An up-transition depends on the label of the
current node and, when the current node is an inner node, on the states reached
at its left and right child.  It defines a new state and the new value of each
counter is the sum of the values of the corresponding counters of the children.
Moreover, the behavior of up-transitions can be modified by the constraints
$\chi$. Informally a constraint of the form $\merc{i_1}{i_2}{i}$ modifies this
process as follows: before performing the addition of the counters, 
two positive numbers
$n_1$ and $n_2$ are guessed 
(possibly of value $0$), the counter $i_1$ of the left child and the counter
$i_2$ of the right child are decreased by $n_1$, the counter~$i_2$ of the left
child and the counter $i_1$ of the right child are decreased by~$n_2$ and, once
the addition of the counters has been executed, the counter~$i$ is increased by
$n_1 + n_2$. Note that $n_1$ and $n_2$ must be so that all intermediate values
remain positive.  This is essentially what is explained in the sketch above
except that we cannot distinguish the left child from the right child. This
will be a property resulting from $\#$-stuttering languages when coding them into binary
trees.
We now make this more precise.

A \emph{configuration} of an \ebvass is a pair $(q,v)$ where $q\in Q$ 
and $v$ is a valuation of the counters, seen as a vector of $\N^k$. 
The initial configuration is $(q_0,v_0)$ where $v_0$ is the function setting all counters to~$0$.
There is an $\epsilon$-\emph{transition} of label $a$ from $(q,v)$ to $(q',v')$ if
$(q,a,q',u) \in \delta_\epsilon$ and $v'=v+u$ 
(in particular this implies that $v+u\geq 0$). 
We write $(q,v) \xrightarrow{a}_\epsilon (q',v')$, if $(q',v')$ can be reached from
$(q,v)$ via a finite sequence of $\epsilon$-transitions of label $a$.
 
Given a binary tree $\atree\in\Trees(\A)$, a \emph{run} $\rho$ of a \ebvass is a
function from nodes of $\atree$ to configurations verifying for all leaf $x$,
$\rho(x)=(q_0,v_0)$ and for all nodes $x,x_1,x_2$ of $\atree$ with $x_1$ and
$x_2$ the left and right child of $x$, and $\rho(x)=(q,v),
\rho(x_1)=(q_1,v_1), \rho(x_2)=(q_2,v_2)$ there exist $(q'_1,v'_1), (q'_2,v'_2)$
such that:

\begin{enumerate}
\item $(q_1,v_1) \xrightarrow{\atree(x_1)}_\epsilon (q'_1,v'_1)$, 
$(q_2,v_2) \xrightarrow{\atree(x_2)}_\epsilon (q'_2,v'_2)$,
\item $(q'_1,\atree(x),q'_2,\; q) \in \delta_{\text{u}}$,
\item for each constraint $\theta \in \chi$ of the form $\merc{i_1}{i_2}{i}$ 
there are two numbers $n^1_\theta$ and $n^2_\theta$ (they may be $0$)
and vectors $u_{\theta,1},u_{\theta,2},u_\theta \in\N^k$, having
$n^1_\theta$ and $n^2_\theta$ at positions $i_1,i_2$ for $u_{\theta,1}$, having
$n^2_\theta$ and $n^1_\theta$ at positions $i_1,i_2$ for $u_{\theta,2}$,
$n^1_\theta+n^2_\theta$ at position $i$ for $u_\theta$
and all other positions set to zero,
\item \label{item:ebvass-decrement}
      $v''_1 = v'_1 - \displaystyle\sum_{\theta\in\chi} u_{\theta,1} \geq 0$, and 
      $v''_2 = v'_2 - \displaystyle\sum_{\theta\in\chi} u_{\theta,2} \geq 0$, and
 $v = v''_1 + v''_2+ \displaystyle\sum_{\theta\in\chi} u_\theta$. 
\end{enumerate}
\noindent
We stress that it will be important for coding the automata into the logic
(Section~\ref{sec-counter-fotwo}) that $\chi$ is independent of the current state of the automaton.

Without the constraints of $\chi$ we have the usual notion of
\bvass~\cite{acl10}. 
It does not seem possible in general to simulate directly a constraint
$\merc{i_1}{i_2}{i}$ with BVASS transitions.  One could imagine using an
arbitrary number of $\epsilon$-transitions decreasing the counters $i_1$ and
$i_2$ while increasing counter $i$, after the merging operation summing up the
counters.  However, it is not clear how to do this while preserving the
positiveness of the corresponding decrements before the merge
(Step~\ref{item:ebvass-decrement} above).

\medskip

\noindent
The \emph{reachability} problem for an \ebvass, on input $q\in Q$, 
asks whether there is a tree and a run on that tree reaching 
the configuration $(q,v_0)$ at its root.

\subsection{Reduction from \texorpdfstring{\dad}{DTA} to \texorpdfstring{\ebvass}{EBVASS}}

\begin{theorem}\label{prop-reduct}
  The emptiness problem for \dad reduces to the reachability problem for
  \ebvass.
\end{theorem}
\begin{proof}
We first take care of the binary trees versus unranked forest issue.
It is well known that forests of $\Forests(\E)$
can be transformed into binary trees in $\Trees(\E_\#)$ 
using the first-child/right-sibling encoding, denoted by $\fcns$,
and formally defined as follows (for $a \in \E$ and $\stree, \stree' \in \Forests(\E)$):
\[
\begin{array}{rcl}
\fcns(a) & = & a(\# + \#)\\
\fcns(a(\stree)) & = & a( \fcns(\stree) + \#)\\
\fcns(a + \stree) & = & a(\# + \fcns(\stree))\\
\fcns(a(\stree) + \stree') & = & a( \fcns(\stree) + \fcns(\stree')).
\end{array}
\]

This transformation effectively preserves regularity: 
for each automaton $\Ba$ computing on $\Forests(\E)$ 
there exists an automaton $\Ba'$ on binary trees of $\Trees(\E_\#)$,
effectively computable from $\Ba$, 
recognizing exactly the $\fcns$ encoding of the forests recognized by $\Ba$.  
This automaton $\Ba'$ is called the $\fcns$ \emph{view} of~$\Ba$.  
Note that we use the same $\#$ symbol in the \fcns construction and 
for class forests. This simplifies the technical details of the proof. In
particular we can assume that our tree automata start with a single initial
state at the leaves of the tree.

\medskip

  We show that given a \dad $\Da$, one can construct an \ebvass $\Ea$ with a
  distinguished state $q$ such that for all $\atree \in \Forests(\A)$, there is
  a run of $\Ea$ on $\fcns(\atree)$ reaching $(q, v_0)$ at its root iff
  $\atree\otimes\dtree$ is accepted by $\Da$ for some~$\dtree$.

\medskip

Before explaining the construction of $\Ea$ we first show the consequences of
the fact that the second component of $\Da$ 
recognizes a $\#$-stuttering language on its \fcns view $\Ba$.
The \fcns view of the rules of Figure~\ref{fig-rules} are depicted in
Figure~\ref{fig-rules-fcns}: One obtains the same result by application of \fcns and then of a rule of Figure~\ref{fig-rules-fcns} 
than by application of the corresponding rule of Figure~\ref{fig-rules} and then of \fcns. 
This can be enforced using the following syntactic restrictions on the \fcns view $\Ba$ that will be useful in our proofs.
In the definition of these restrictions, 
we use the notation $(p_1,b,p_2) \ra p$ 
for a transition of $\Ba$ from the states $p_1$, $p_2$ 
in the left and right child of a node of label $b$, moving up with state $p$.

We assume without loss of generality that the states of $\Ba$ permit to
distinguish the last symbol read by $\Ba$.  More precisely, we assume that the
set of states of $\Ba$ is split into two kinds: the $\#$-states and the
non-$\#$-states.  The states of the first kind are reached by $\Ba$ on nodes
labeled with symbol $\#$, while the states of the second kind are reached by
$\Ba$ on nodes with label in $\B$.  We say that $\Ba$ is
\emph{$\#$-stuttering} if $\Ba$ is deterministic and has a specific $\#$-state
$p_\#$ that it must reach on all leaves of label $\#$, and verifies the
following properties:
\begin{enumerate} 
\item if a transition rule of the form $(p_1, \#, p_\#) \to p_2$ 
is applied at a $\#$-node that is the left-child of another $\#$-node, then $p_1=p_2$
\item if a transition rule of the form $(p_\#, \#, p_1) \to p_2$ is applied at a $\#$-node that is the right-child of another $\#$-node, then $p_1=p_2$
\item all transition rules of the form $(p_\#, \#, p_1) \to p_2$ with $p_1$ a $\#$-state
  verify $p_1=p_2$.
\item \label{it-commutativity}
 all transition rules of the form $(p_1, \#, p_2) \to p$ are ``commutative'',
  \textit{i.e.}, $(p_2, \#, p_1) \to p$ must then also be a rule.
\end{enumerate}

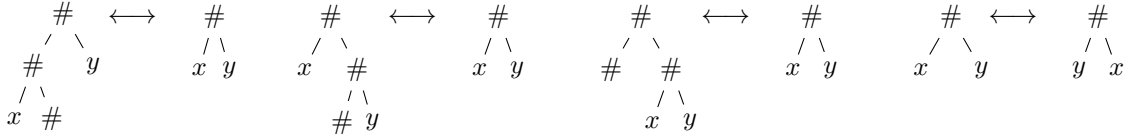
\begin{figure}
\small
\[
\begin{array}{rclcrclcrclrclc}
\begin{minipage}[t]{7mm}
\vspace{-4mm}
\begin{tikzpicture}
  [level distance=7mm,
   level 1/.style={sibling distance=8mm},
   level 2/.style={sibling distance=5mm}]
\node {$\#$}
 child { node {$\#$}
         child { node {$x$} }
         child { node {$\#$} }
       }
 child { node {$y$} };
\end{tikzpicture}
\end{minipage}
&
\vspace{-0mm}\,\,\,\,\,\,\,\toto
&
\begin{minipage}[t]{7mm}
\vspace{-3.5mm}
\begin{tikzpicture}
  [level distance=7mm,
   level 1/.style={sibling distance=4mm}]
\node {$\#$}
 child { node {$x$} }
 child { node {$y$} };
\end{tikzpicture}
\end{minipage}
& & 
\begin{minipage}[t]{8mm}
\vspace{-3.5mm}
\begin{tikzpicture}
  [level distance=7mm,
   level 1/.style={sibling distance=7mm},
   level 2/.style={sibling distance=4mm}]
\node {$\#$}
 child { node {$x$} }
 child { node {$\#$}
         child { node {$\#$} }
         child { node {$y$} }
       };
\end{tikzpicture}
\end{minipage}
&
\vspace{-0mm}\,\,\,\toto
&
\begin{minipage}[t]{10mm}
\vspace{-3.5mm}
\begin{tikzpicture}
  [level distance=7mm,
   level 1/.style={sibling distance=5mm}]
\node {$\#$}
 child { node {$x$} }
 child { node {$y$} };
\end{tikzpicture}
\end{minipage}
& & 
\begin{minipage}[t]{10mm}
\vspace{-3.5mm}
\begin{tikzpicture}
  [level distance=7mm,
   level 1/.style={sibling distance=8mm},
   level 2/.style={sibling distance=5mm}]
\node {$\#$}
 child { node {$\#$} }
 child { node {$\#$}
         child { node {$x$} }
         child { node {$y$} }
       };
\end{tikzpicture}
\end{minipage}
&
\vspace{-0mm}\,\,\toto
&
\begin{minipage}[t]{10mm}
\vspace{-3.5mm}
\begin{tikzpicture}
  [level distance=7mm,
   level 1/.style={sibling distance=5mm}]
\node {$\#$}
 child { node {$x$} }
 child { node {$y$} };
\end{tikzpicture}
\end{minipage}
& & 
\begin{minipage}[t]{10mm}
\vspace{-3.5mm}
\begin{tikzpicture}
  [level distance=7mm,
   level 1/.style={sibling distance=8mm},
   level 2/.style={sibling distance=5mm}]
\node {$\#$}
 child { node {$x$} }
 child { node {$y$}};
\end{tikzpicture}
\end{minipage}
&
\vspace{-0mm}\!\!\!\!\toto
&
\begin{minipage}[t]{10mm}
\vspace{-3.5mm}
\begin{tikzpicture}
  [level distance=7mm,
   level 1/.style={sibling distance=5mm}]
\node {$\#$}
 child { node {$y$} }
 child { node {$x$} };
\end{tikzpicture}
\end{minipage}
\end{array}
\]
\caption{\fcns view of the $\#$-stuttering closure rules. $x$ and $y$ are
  arbitrary binary trees.}
\label{fig-rules-fcns}
\end{figure}

\noindent From these definitions, it is straightforward to see that 
for a set $L \subseteq \Forests(\B)$, the following properties are equivalent
\begin{itemize}
\item $L$ is an $\#$-stuttering language,
\item $\fcns(L)$ is closed under the rules in Figure~\ref{fig-rules-fcns},
\item there exists an $\#$-stuttering automaton recognizing $\fcns(L)$.
\end{itemize}

\medskip

We now turn to the construction of $\Ea$.
Let $\Aa = (Q_\Aa, \A, \B, q^0_{\Aa}, F_\Aa, \Delta_\Aa)$ 
and $\Ba = (Q_\Ba, \B_\#, q^0_{\Ba}, F_\Ba, \Delta_\Ba)$
be the $\fcns$ views of the two components of $\Da$.
The automaton $\Ba$ is $\#$-stuttering (\textit{i.e.} there is a distinction in its
states between $\#$-states and non-$\#$-states, and the existence of a
$\#$-state $p_\# \in Q_\Ba$ on which $\Ba$ evaluates the tree with a single
node labeled with $\#$) and we also assume without loss of generality that it
is deterministic and complete, \textit{i.e.}, for every $\btree \in
\Trees(\B_\#)$, $\Ba$ evaluates into exactly one state of $Q_\Ba$.

Here $Q_\Aa$ and $Q_\Ba$ (resp. $F_\Aa$, $F_\Ba$) are the respective state sets
(resp. final state sets) of $\Aa$ and $\Ba$, 
$\A$ is the input alphabet of $\Aa$, 
$\B$ is the output alphabet of $\Aa$ and input alphabet of $\Ba$
(with the symbol~$\#$),
and $\Delta_\Aa$, $\Delta_\Ba$ are the sets of transitions.
We will use the notation  
$(r_1, a, r_2,\; r, b)$,
for a transition of $\Aa$ from the states $r_1$, $r_2$ in the left and right child 
of a node of label $a$, renaming this node with $b$ and moving up with state $r$.
In the following, we write explicitly the set of states of $\Ba$ 
as $Q_\Ba = \{ p_\#,p_1,\ldots,p_k \}$.

For any data tree $\tree\in\Trees(\B_\# \times\D)$, and any data value $d$ occurring in $\tree$, 
the state of $Q_\Ba$ corresponding to the evaluation of $\Ba$ on the class forest $\tree[d]$
is called \emph{the $\Ba$-state associated to $d$ in $\tree$}. 
When $d$ is the data value at the root of $\tree$,
this state is called the \emph{\rootstate{}} of $\tree$.
Note that for all $\tree$ the \rootstate{} of $\tree$ exists and is unique,
since $\Ba$ is assumed to be deterministic and complete, 
and that it is always a non-$\#$-state.

We now construct the expected \ebvass $\Ea = (Q,\A,q_0,k,\delta,\chi)$ 
with $k = |Q_\Ba|-1$.
We set $Q = Q_\Aa \times Q_\Ba \times Q_0$, 
where $Q_0$ is a finite set of auxiliary control states.
The initial state $q_0$ is the tuple formed with $q^0_{\Aa}$, the initial state of $\Aa$,
$q^0_{\Ba}$, the initial state of $\Ba$, and a specific state of $Q_0$.
The first and second components of a state $q\in Q$
are respectively called the $\Aa$-state and the 
\rootstate of $q$.
The transitions of the \ebvass $\Ea$ are constructed in order to ensure
the following invariant:

\medskip
\noindent
\hspace{-1mm}$(\star)$%
\hspace{2mm}\parbox[t]{128mm}{$\Ea$ reaches the configuration $(q,v)$ at the
  root of a tree $\atree \in \Trees(\A)$ 
  iff there exists a data tree $\tree=\atree\otimes\dtree$
  and a possible output $\btree \in \Trees(\B)$ of $\Aa$ on $\atree$ 
  witnessed by a run of $\Aa$ whose state at the root of $\atree$ is the $\Aa$-state of~$q$,   
  and moreover for all $i$, $1 \leq i \leq k$,
  $v_i$ is the number of data values having $p_i$ as associated $\Ba$-state 
  in $\btree\otimes\dtree$.}  
\medskip

\noindent Note that the counters ignore the number of data values having $p_\#$ as associated $\Ba$-state 
(which will always be infinite).
A consequence of $(\star)$ 
is that:
$(\diamond)$ there is only one non-$\#$-state $p_i \in Q_\Ba$ such that $v_i \neq 0$, and actually $v_i = 1$.
We will refer to this state $p_i$ as the \emph{\rootstate} of $v$,
and the construction of $\Ea$ will ensure that $p_i$ is also the \rootstate of $q$.

If we can achieve the invariant $(\star)$ then we are done.
Indeed, we can add to $\Ea$ some $\epsilon$-transitions which, 
when reaching a state $q$ containing a final $\Aa$-state,
decrement the counters corresponding to final states of $\Ba$ (and only those).
Then, $\Ea$ reaches a configuration $(q,v_0)$ with the $\Aa$-state of $q$ being
a final state of $\Aa$ iff there exists a data tree accepted by~$\Da$.

Notice that the property $(\star)$ is invariant under permutations of $\D$.
Hence if a tree $\dtree$ witnesses the property $(\star)$, then any tree
$\dtree'$ constructed from $\dtree$ by permuting the data values is also a
witness for $(\star)$.  This observation will be useful for showing the
correctness of the construction of $\Ea$.

Before defining the transition relation of $\Ea$ we sketch with more details its
construction.

The automaton $\Ea$ needs to maintain the invariant $(\star)$.  One direction
will be immediate: if $\Da$ has an accepting run on $\atree\otimes\dtree$ then
$\Ea$ is constructed so that it has an accepting run on $\atree$ satisfying
$(\star)$ as witnessed by $\dtree$. For the converse direction, we need to
construct from a run of $\Ea$ on $\atree$ a tree $\dtree$ such that
$\Da$ has a run on $\atree \otimes \dtree$ as in $(\star)$.

The simulation of $\Aa$ is straightforward as $\Ea$ can simulate any regular
tree automaton. The simulation of $\Aa$ is done using the $\Aa$-state of the
states of $\Ea$: for every $\epsilon$-transition $(q,a,\; q',u)$ of $\Ea$, the
$\Aa$-states of $q$ and $q'$ must coincide, and for every up-transition
$(q_1, a, q_2,\; q)$ of $\Ea$, there exists a transition of $\Aa$ of the form
$(r_1, a, r_2,\; r, b)$, for some $b \in \B$ such that $r_1,r_2$ and $r$ are
the respective $\Aa$-states of $q_1$, $q_2$ and $q$. In other words, the
$\Aa$-state of $\Ea$ always is the state of $\Aa$ at the current node.

Let's now turn to the simulation of $\Ba$ and the invariant $(\star)$.
This invariant will be shown by induction in the depth of the tree.
Let us assume that $\Ea$ reached the configuration $(q,v)$ 
at the root $x$ of a tree $\atree$.

If $x$ is a leaf node, then by definition of \ebvass, 
$q$ is the initial state $q_0$ of $\Ea$ 
and $v = v_0$,
hence $(\star)$ holds.

If $x$ is an inner node, then $\atree=a(\atree_1,\atree_2)$ for some letter
$a\in \A$.  By induction on the depth, we have trees $\dtree_1$
and $\dtree_2$ such that there is a run of $\Da$ on $\atree_1\otimes\dtree_1$ and $\atree_2\otimes\dtree_2$ satisfying $(\star)$. 
From the remark above on the invariance of $(\star)$ under permutations of $\D$, 
we can assume that $\dtree_1$ and $\dtree_2$
do not share any data value.
We need to set the transitions of $\Ea$ such that from $\dtree_1$, $\dtree_2$,
we can construct a tree $\dtree$ such that $\Da$ also has 
a run on $\atree\otimes\dtree$ as in ($\star$). The tree $\dtree$ will
be of the form $d(\dtree'_1,\dtree_2')$ for some $d\in \D$, where $\dtree'_1$
and $\dtree'_2$ are constructed from $\dtree_1$ and $\dtree_2$ by permuting
the data values. The permutation will identify some data values of $\dtree_1$
with some data values of $\dtree_2$. The number of data values we identify is
given by the $n$ in the constraints of $\Ea$ as explained in the initial
sketch on page~\pageref{sketch:sec-dad-counter}. 
This $n$ is therefore given by the run of $\Ea$ and we will see that it does not
matter which data values we actually choose, it is only important that we pick
$n$ of them. The constraints make sure that this is consistent with the runs of $\Ba$.

For this purpose we define $\chi$ as the set of constraints of the form
$\merc{j_1}{j_2}{j}$ such that there exists a transition $(p_{j_1},\#, p_{j_2})
\to p_j$ of $\Ba$ where $p_{j_1}$, $p_{j_2}$, and $p_j$ are $\#$-states in
$Q_\Ba \setminus \{ p_\#\}$.  Note that the commutativity rule in the
definition of $\#$-stuttering languages implies that whenever we have a
constraint $\merc{j_1}{j_2}{j}$ then both $(p_{j_1},\#, p_{j_2}) \to p_j$ and
$(p_{j_2},\#, p_{j_1}) \to p_j$ are transitions of $\Ba$.

This does maintain $(\star)$ assuming that $d$, the data value expected at $x$,
is not among the data values we identify (in the transitions used to construct
$\chi$ the root symbol is $\#$). This data value $d$ has to be treated
separately and we have several cases depending on whether $d$ is completely new
(does not occur in~$\dtree'_1$ and~$\dtree'_2$), or occurs in~$\dtree'_1$ but
not in~$\dtree'_2$, or the other way round, or it occurs in both subtrees. Actually it
will also be necessary to consider separately the cases where~$d$ occurs at the
root of~$\dtree_1'$ or~$\dtree'_2$.

This last choice is guessed by $\Ea$ and can therefore be read from the run
of $\Ea$. We can then choose $d$ consistently with the guess of $\Ea$. Again
the precise value of $d$ is not important. It is only important that its
equality type with the other data values is consistent with the choice made by
$\Ea$. This makes finitely many cases and we define the transition function of
$\Ea$ as the union of corresponding family of transitions. Each of them
involving disjoint intermediate states they don't interfere between each
other. We therefore define them separately and immediately after prove that they
do maintain $(\star)$ for their case.

\medskip

\paragraph{\bf 1.}
\emph{$\Ea$ guessed that the data value of the current node 
 is equal to the data value of both its children}.

\noindent
To handle this case, for each transition $\tau=(p_{i_1},b,p_{i_2}) \ra p_{i}$ of $\Ba$, 
where none of $p_{i_1},p_{i_2},p_i$ are $\#$-states, $\Ea$ has the following transitions:
\begin{description}
\item[\rm $\epsilon$-transitions:] \quad\\
from a state $q_1$ of $\Ba$-state $p_{i_1}$ 
it decreases counter $i_1$ 
and moves to a state $q_\tau^1$\\ 
from a state $q_2$ of $\Ba$-state $p_{i_2}$
it decreases counter $i_2$ 
and moves to state $q_\tau^2$\\
from a state $q_\tau$ 
it increases counter $i$ 
and moves to a state $q$ of $\Ba$-state $p_i$

\item[\rm up-transition:]
$(q_\tau^1,a,q_\tau^2,q_\tau)$.
\end{description}
The state $q_\tau^1$ (resp. $q_\tau^2$, $q_\tau$) 
differs from $q_1$  (resp. $q_2$, $q$)
only by its third component (in $Q_0$), that contains $\tau$.
We shall use the same convention for the states introduced in the 
following construction cases.

\paragraph{\bf Correctness.} 
Let us show that if $\Ea$ makes an up-transition $(q^1_\tau,a,q^2_\tau,q_\tau)$  
at the root of $\atree \in \Trees(\A)$ we can construct $\dtree$ such that $\Da$
has a run on $\atree\otimes\dtree$ satisfying~$(\star)$.
This up-transition can only occur if we had $\epsilon$-transitions from $q_1$ to
$q_\tau^1$ in the left subtree and from $q_2$ to
$q_\tau^2$ in the right subtree where $p_{i_1}$ and $p_{i_2}$ are the
$\Ba$-states of $q_1$ and $q_2$.
Let $x$ be the root of the tree $\atree\otimes\dtree$ where this transition
occurred. We have $\atree = a(\atree_1,\atree_2)$. 
By induction hypothesis we have trees
$\dtree_1$ and $\dtree_2$ 
and possible outputs $\btree_1, \btree_2 \in \Trees(\B)$ of $\Aa$ 
on respectively $\atree_1$ and $\atree_2$ 
such that there is a run of $\Da$ on 
$\tree_1 = \atree_1\otimes\dtree_1$ and $\tree_2 = \atree_2\otimes\dtree_2$ satisfying~$(\star)$.

We first apply a bijection on the labels of $\dtree_1$ in order for the data
value of its root to match the one of the root of $\dtree_2$. Let $d$ be this
data value.

For each constraint $\theta=\merc{k_1}{k_2}{k} \in \chi$ 
we let $n^1_\theta$ and $n^2_\theta$ be the numbers used by the run of $\Ea$ when using
the above up-transition.  By induction hypothesis $(\star)$, and semantics of the
constraints (making sure the counters are big enough) there are at least
$n^1_\theta$ (resp. $n^2_\theta$) distinct data values different from $d$
(because the up-transition is applied \emph{after} we decreased the counter $k_1$ by $n^1_\theta$) 
in $\dtree_1$ having $p_{k_1}$ (resp. $p_{k_2}$) as associated
$\Ba$-state in $\tree_1 = \btree_1\otimes\dtree_1$, 
and similarly for $\tree_2 = \btree_2\otimes\dtree_2$.  
We pick such data values in each subtree 
and call them the data values associated to $\theta$.  
We do this for all constraints $\theta$ 
and we choose the associated data values such that they are all distinct.  
We now apply to $\dtree_2$ a permutation on the data values 
such that for all $\theta$ the data values associated to $\theta$ in $\dtree_2$ 
are identified with the ones for $\dtree_1$ and such that all other
data values are distinct.  In order to simplify the notations we call the
resulting tree also $\dtree_2$.  We then set $\dtree$ as $d(\dtree_1,\dtree_2)$
and $\tree'=\btree\otimes\dtree$ where $\btree = b(\btree_1, \btree_2)$ 
is an output of $\Aa$ on $\atree$ compatible with the transition.

Let $e$ be an arbitrary data value occurring in $\dtree$.

If $e = d$, the root symbol of the class forest $\tree'[e]$ 
is $a$ and the counter $i$ is increased by 1 by the last $\epsilon$-transition.
By induction hypothesis and its consequence $(\diamond)$, 
$v_i = 1$ and for all other non-$\#$-states 
the corresponding value via $v$ will be 0. 
Hence $p_i$ is the new \rootstate of $v$. 
It is also the \rootstate of $q$ by construction.

If $e \neq d$ we consider 3 subcases. 
If $e$ occurs in both $\dtree_1$ and $\dtree_2$ 
then the class forest $\tree'[e]$ has the form $\#(\stree_1,\stree_2)$ 
for some forests $\stree_1$ and $\stree_2$
containing each at least one symbol other than $\#$ (not at the root node).
Let $p_{j_1}$ and $p_{j_2}$ be the states reached by $\Ba$ when evaluating $\stree_1$ and $\stree_2$.
They are the $\Ba$-states associated to $e$ in $\tree'_1$ and $\tree'_2$,
(resp. the left- and right subtrees of $\tree'$),
and both are $\#$-states in $Q_\Ba \setminus \{ p_\# \}$.
By construction of $\dtree$, there are at least
$n_\theta=n_\theta^1+n_\theta^2$ such data values $e$, 
where $\theta = \merc{j_1}{j_2}{j}$ 
and $p_j$ is the unique state of $\Ba$ such that 
$(p_{j_1}, \#, p_{j_2}) \to p_{j}$ is a transition of $\Ba$ (and
therefore also $(p_{j_2}, \#, p_{j_1}) \to p_j$ is also a transition). 
These $n_\theta$ data values will contribute to an increase of $v_j$ 
by $n_\theta$ as expected.

Assume now that $e$ occurs in $\dtree_1$ but not in $\dtree_2$ (the remaining
case being symmetrical). Then $\tree'[e]$ has the form $\#(\stree_1,\stree_2)$ 
where $\stree_1$ contains at least one symbol other than $\#$ (not at root node),
and all nodes of $\stree_2$ are labeled~$\#$.
By the hypothesis that $\Ba$ is $\#$-stuttering, 
the $\Ba$-state associated to $e$ in $\tree'$ is the same as the one
associated to $e$ in $\tree'_1$,
and the $\Ba$-state associated to $e$ in $\tree'_2$ is $p_\#$.
This is consistent with the behavior of $\Ea$
that propagates upward the value of the counter corresponding to this state, 
after applying the constraints.
Altogether this shows that $\tree = \atree\otimes\dtree$ verifies $(\star)$.

\tikzset{
itria/.style={
  draw,shape border uses incircle,
  isosceles triangle,shape border rotate=90,yshift=-0.7cm},
rtria/.style={
  draw,dashed,shape border uses incircle,
  isosceles triangle,isosceles triangle apex angle=90,
  shape border rotate=-45,yshift=0.2cm,xshift=0.5cm},
ritria/.style={
  draw,dashed,shape border uses incircle,
  isosceles triangle,isosceles triangle apex angle=110,
  shape border rotate=-55,yshift=0.1cm},
letria/.style={
  draw,dashed,shape border uses incircle,
  isosceles triangle,isosceles triangle apex angle=110,
  shape border rotate=235,yshift=0.1cm}
}

\begin{figure}
\small
\[
\begin{array}{ccc}
\tree'=\btree\otimes\dtree = 
\hspace{-8mm}
\begin{minipage}[t]{24mm}
\vspace{-3.5mm}
\begin{tikzpicture}
  [level distance=7mm,
   level 1/.style={sibling distance=12mm},
   level 2/.style={sibling distance=4mm}]
\node {$(b, d)$} 
 child { node {$(b_1, d)$} { node[itria] {} } }
 child { node {$(b_2, d)$} { node[itria] {} } };
\end{tikzpicture}
\end{minipage}
& 
\tree'[d] = 
\hspace{-8mm}
\begin{minipage}[t]{24mm}
\vspace{-3.5mm}
\begin{tikzpicture}
  [level distance=7mm,
   level 1/.style={sibling distance=12mm},
   level 2/.style={sibling distance=4mm}]
\node {$b (p_i)$} 
 child { node {$b_1 (p_{i_1})$} { node[itria] {} } }
 child { node {$b_2 (p_{i_2})$} { node[itria] {} } };
\end{tikzpicture}
\end{minipage}
\end{array}
\]
\caption{Proof of Theorem~\ref{prop-reduct}, Case 1.
The $\Ba$-states are displayed in parentheses in the class tree $\tree'[d]$.}
\label{fig-reduct1}
\end{figure}
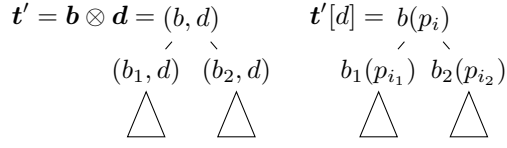

\paragraph{\bf 2.}  \emph{$\Ea$ guessed that the data value $d_1$ of
  the current node is equal to the data value of its left child but different
  from the data value $d_2$ of its right child. 
  Moreover $\Ea$ guessed that the $\Ba$-state associated to $d_1$ in the right
  subtree is $p_{k_2}$, and that the data value $d_2$ of the right child 
  also appears in the left subtree, with $p_{k_1}$ as $\Ba$-state associated to $d_2$ in this left subtree.
  Note that both $p_{k_1}$ and $p_{k_2}$ must be $\#$-states in $Q_\Ba \setminus \{ p_\#\}$.}

\noindent
To handle this case for all transitions $\tau=(p_{i_1},b,p_{k_2}) \ra p_{i}$
and $\tau'=(p_{k_1},\#,p_{i_2})\ra p_j$ of
$\Ba$, where none of $p_{i_1},p_{i_2}, p_i$ are $\#$-states but
$p_j$ (like $p_{k_1}$ and $p_{k_2}$) are $\#$-states, 
$\Ea$ has the following transitions:
\begin{description}
\item[\rm $\epsilon$-transitions:] \quad\\
from a state $q_1$ of $\Ba$-state $p_{i_1}$
it decreases the counters $i_1$ and $k_1$ 
and moves to state $q_{\tau,\tau'}^1$\\
from a state $q_2$ of $\Ba$-state $p_{i_2}$ 
it decreases the counters $i_2$ and $k_2$ 
and moves to state $q_{\tau,\tau'}^2$\\
from a state $q_{\tau,\tau'}$ 
it increases counters $i$ and $j$ 
and moves to a state $q$ of $\Ba$-state $p_i$

\item[\rm up-transition:]
$(q_{\tau,\tau'}^1,a,q_{\tau,\tau'}^2,q_{\tau,\tau'})$.
\end{description}

\paragraph{\bf Correctness.}  
We argue as in the previous case with the following modifications.
From $\dtree_1$ and $\dtree_2$ we first apply a bijection making sure that the
data values $d_1$ and $d_2$ of their roots are different and that $d_1$ has
$\Ba$-state $p_{k_2}$ in $\btree_2 \otimes \dtree_2$ 
and $d_2$ has $\Ba$-state $p_{k_1}$ in $\btree_1 \otimes \dtree_1$,
where $\atree = a(\atree_1,\atree_2)$ and 
$\btree_1, \btree_2 \in \Trees(\B)$ are possible outputs of $\Aa$ 
on respectively $\atree_1$ and $\atree_2$ from the induction hypothesis.

For each $\theta \in \chi$ we select the associated data values making sure
they are neither $d_1$ nor $d_2$. The decrement in the $\epsilon$-transitions make sure
that this is always possible. We then perform the same identification as in the
previous case. 
The same argument as above shows that the resulting tree $\dtree=d_1(\dtree_1,\dtree_2)$ has the desired properties.

\begin{figure}
\small
\[
\begin{array}{ccc}
\tree'=\btree\otimes\dtree = 
\hspace{-8mm}
\begin{minipage}[t]{24mm}
\vspace{-3.5mm}
\begin{tikzpicture}
  [level distance=7mm,
   level 1/.style={sibling distance=12mm},
   level 2/.style={sibling distance=4mm}]
\node {$(b, d_1)$} 
 child { node {$(b_1, d_1)$} { node[itria] {} } }
 child { node {$(b_2, d_2)$} { node[itria] {} } };
\end{tikzpicture}
\end{minipage}
& 
\tree'[d_1] = 
\hspace{-8mm}
\begin{minipage}[t]{24mm}
\vspace{-3.5mm}
\begin{tikzpicture}
  [level distance=7mm,
   level 1/.style={sibling distance=12mm},
   level 2/.style={sibling distance=4mm}]
\node {$b (p_i)$} 
 child { node {$b_1 (p_{i_1})$} { node[itria] {} } }
 child { node {$\# (p_{k_2})$} { node[itria] {} } };
\end{tikzpicture}
\end{minipage}
& 
\tree'[d_2] = 
\hspace{-8mm}
\begin{minipage}[t]{24mm}
\vspace{-3.5mm}
\begin{tikzpicture}
  [level distance=7mm,
   level 1/.style={sibling distance=12mm},
   level 2/.style={sibling distance=4mm}]
\node {$\# (p_j)$} 
 child { node {$\# (p_{k_1})$} { node[itria] {} } }
 child { node {$b_2 (p_{i_2})$} { node[itria] {} } };
\end{tikzpicture}
\end{minipage}
\end{array}
\]
\caption{Proof of Theorem~\ref{prop-reduct}, Case 2.}
\label{fig-reduct2}
\end{figure}

\paragraph{\bf 3.}  
\emph{$\Ea$ guessed that the data value $d_1$ of the
  current node is equal to the data value of its left child but different
  from the data value 
  of its right child.
  Moreover $\Ea$ guessed that $d_1$ also appear in
  the right subtree of the current node, with $p_{k_2}$ as associated $\Ba$-state in this right subtree, 
  and that the data value 
  of the right child of the current node does not appear in the left subtree. 
  Note that $p_{k_2}$ must be a $\#$-state in $Q_\Ba \setminus \{ p_\#\}$.}

\noindent
To handle this case for all transitions $\tau=(p_{i_1},b,p_{k_2}) \ra p_{i}$
and $\tau'=(p_{\#},\#,p_{i_2})\ra p_j$ of
$\Ba$, where none of $p_{i_1},p_{i_2}, p_i$ are $\#$-states but
$p_{k_2}$ and $p_j$  are $\#$-states, 
$\Ea$ has the following transitions:
\begin{description}
\item[\rm $\epsilon$-transitions:] \quad\\
from a state $q_1$ of $\Ba$-state $p_{i_1}$
it decreases the counter $i_1$ 
and moves to state $q_{\tau,\tau'}^1$\\
from a state $q_2$ of $\Ba$-state $p_{i_2}$ 
it decreases the counter $i_2$ 
and $k_2$ and moves to state $q_{\tau,\tau'}^2$\\
from a state $q_{\tau,\tau'}$ 
it increases the counters $i$ and $j$ 
and moves to a state $q$ of $\Ba$-state $p_i$

\item[\rm up-transition:]
$(q_{\tau,\tau'}^1,a,q_{\tau,\tau'}^2,q_{\tau,\tau'})$.
\end{description}

\paragraph{\bf Correctness.}  
We argue as in the previous cases with the following modifications.
From $\dtree_1$ and $\dtree_2$ we first apply a bijection making sure that the
data values $d_1$ and $d_2$ of their roots are different and that $d_1$ has
$\Ba$-state $p_{k_2}$ in $\btree_2 \otimes \dtree_2$ 
and $d_2$ does not appear in $\dtree_1$ ($\btree_2$ is as in previous cases).

For each $\theta \in \chi$ we select the associated data values making sure
they are neither $d_1$ nor $d_2$. The decrement in the $\epsilon$-transitions make sure
that this is always possible. We then perform the same identification as in the
previous case. 
As before we show that the resulting tree $\dtree=d_1(\dtree_1,\dtree_2)$ has the desired properties.

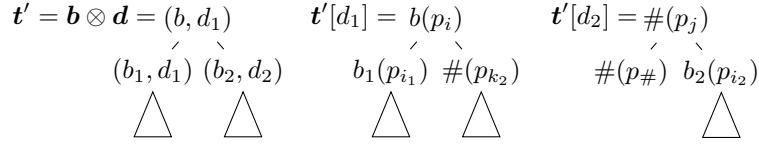
\begin{figure}
\small
\[
\begin{array}{ccc}
\tree'=\btree\otimes\dtree = 
\hspace{-8mm}
\begin{minipage}[t]{24mm}
\vspace{-3.5mm}
\begin{tikzpicture}
  [level distance=7mm,
   level 1/.style={sibling distance=12mm},
   level 2/.style={sibling distance=4mm}]
\node {$(b, d_1)$} 
 child { node {$(b_1, d_1)$} { node[itria] {} } }
 child { node {$(b_2, d_2)$} { node[itria] {} } };
\end{tikzpicture}
\end{minipage}
& 
\tree'[d_1] = 
\hspace{-8mm}
\begin{minipage}[t]{24mm}
\vspace{-3.5mm}
\begin{tikzpicture}
  [level distance=7mm,
   level 1/.style={sibling distance=12mm},
   level 2/.style={sibling distance=4mm}]
\node {$b (p_i)$} 
 child { node {$b_1 (p_{i_1})$} { node[itria] {} } }
 child { node {$\# (p_{k_2})$} { node[itria] {} } };
\end{tikzpicture}
\end{minipage}
& 
\tree'[d_2] = 
\hspace{-8mm}
\begin{minipage}[t]{24mm}
\vspace{-3.5mm}
\begin{tikzpicture}
  [level distance=7mm,
   level 1/.style={sibling distance=12mm},
   level 2/.style={sibling distance=4mm}]
\node {$\# (p_j)$} 
 child { node {$\# (p_\#)$} }
 child { node {$b_2 (p_{i_2})$} { node[itria] {} } };
\end{tikzpicture}
\end{minipage}
\end{array}
\]
\caption{Proof of Theorem~\ref{prop-reduct}, Case 3. The node without subtree in the class tree $\tree'[d_2]$ is a leaf.}
\label{fig-reduct3}
\end{figure}

\paragraph{\bf 4.} 
\emph{$\Ea$ guessed that the data value $d$ of the
  current node is different from the ones of its children but appear in both
  subtrees, with $p_{k_1}$ and $p_{k_2}$ as associated $\Ba$-states repectively in left and right subtrees. 
  Moreover $\Ea$ guessed that the data values of both children of the current node 
  are equal. 
  Note that $p_{k_1}$ and $p_{k_2}$ must be \#-states in $Q_\Ba \setminus \{ p_\#\}$.}

\noindent
To handle this case for all transitions $\tau=(p_{k_1},b,p_{k_2}) \ra p_{i}$
and $\tau'=(p_{i_1},\#,p_{i_2})\ra p_j$ of
$\Ba$, where none of $p_{i_1},p_{i_2}, p_i$ are $\#$-states but
$p_{k_1},p_{k_2}$ and $p_j$  are $\#$-states, $\Ea$ has the following transitions:
\begin{description}
\item[\rm $\epsilon$-transitions:]\quad\\
from a state $q_1$ of $\Ba$-state $p_{i_1}$ 
it decreases the counters $i_1$ and $k_1$ 
and moves to state $q_{\tau,\tau'}^1$\\
from a state $q_2$ of $\Ba$-state $p_{i_2}$ it
decreases the counters $i_2$ and $k_2$ 
and moves to state $q_{\tau,\tau'}^2$\\
from a state $q_{\tau,\tau'}$ 
it increases the counters $i$ and $j$ 
and moves to a state $q$ of $\Ba$-state $p_i$

\item[\rm up-transition:]
$(q_{\tau,\tau'}^1,a,q_{\tau,\tau'}^2,q_{\tau,\tau'})$.
\end{description}

\paragraph{\bf Correctness.}  
We argue as in the previous cases with the following modifications.

From $\dtree_1$ and $\dtree_2$ we first apply a bijection making sure that the
data value $d_1$ of their roots are equal and that $\dtree_1$ and $\dtree_2$ share a common
data value $d\neq d_1$ of $\Ba$-state $p_{k_2}$ in $\btree_2 \otimes\dtree_2$ 
and $\Ba$-state $p_{k_1}$ in $\btree_1 \otimes \dtree_1$.

For each $\theta \in \chi$ we select the associated data values making sure
they are neither $d_1$ nor $d$. The decrement in the $\epsilon$-transitions make sure
that this is always possible. We then perform the same identification as in the
previous case. 
The rest of the argument is similar after setting $\dtree=d(\dtree_1,\dtree_2)$.

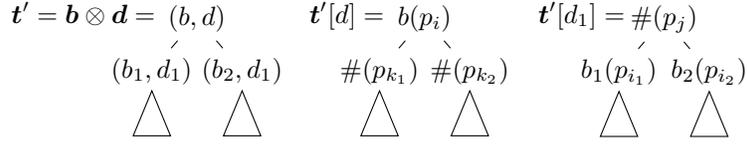
\begin{figure}
\small
\[
\begin{array}{ccc}
\tree'=\btree\otimes\dtree = 
\hspace{-8mm}
\begin{minipage}[t]{24mm}
\vspace{-3.5mm}
\begin{tikzpicture}
  [level distance=7mm,
   level 1/.style={sibling distance=12mm},
   level 2/.style={sibling distance=4mm}]
\node {$(b, d)$} 
 child { node {$(b_1, d_1)$} { node[itria] {} } }
 child { node {$(b_2, d_1)$} { node[itria] {} } };
\end{tikzpicture}
\end{minipage}
& 
\tree'[d] = 
\hspace{-8mm}
\begin{minipage}[t]{24mm}
\vspace{-3.5mm}
\begin{tikzpicture}
  [level distance=7mm,
   level 1/.style={sibling distance=12mm},
   level 2/.style={sibling distance=4mm}]
\node {$b (p_i)$} 
 child { node {$\# (p_{k_1})$} { node[itria] {} } }
 child { node {$\# (p_{k_2})$} { node[itria] {} } };
\end{tikzpicture}
\end{minipage}
& 
\tree'[d_1] = 
\hspace{-8mm}
\begin{minipage}[t]{24mm}
\vspace{-3.5mm}
\begin{tikzpicture}
  [level distance=7mm,
   level 1/.style={sibling distance=12mm},
   level 2/.style={sibling distance=4mm}]
\node {$\# (p_j)$} 
 child { node {$b_1 (p_{i_1})$} { node[itria] {} } }
 child { node {$b_2 (p_{i_2})$} { node[itria] {} } };
\end{tikzpicture}
\end{minipage}
\end{array}
\]
\caption{Proof of Theorem~\ref{prop-reduct}, Case 4.}
\label{fig-reduct4}
\end{figure}

\medskip\noindent\textbf{5.}  \emph{$\Ea$ guessed that the data value $d$ of the
  current node is different from the ones of its children but appear in
  both subtrees, with $p_{k_1}$ and $p_{k_2}$ as associated $\Ba$-states in respectively the left and right subtree.
  Moreover $\Ea$ guessed that the data values 
  of both children of the current node are
  distinct but appear in the other subtree with respective associated $\Ba$-state $p_{\ell_1}$ and $p_{\ell_2}$.
  Note that $p_{k_1}$, $p_{k_2}$, $p_{\ell_1}$, $p_{\ell_2}$ must be \#-states.}

To handle this case for all transitions
$\tau = (p_{k_1},b,p_{k_2}) \ra p_{i}$,
$\tau_1=(p_{i_1},\#,p_{\ell_1})\ra p_{j_1}$ and 
$\tau_2=(p_{\ell_2},\#,p_{i_2})\ra p_{j_2}$ of $\Ba$, 
where none of $p_{i_1},p_{i_2}, p_i$ are $\#$-states but
$p_{k_1},p_{k_2},p_{\ell_1},p_{\ell_2},p_{j_1},p_{j_2}$ are $\#$-states, $\Ea$ has the following
transitions:
\begin{description}
\item[\rm $\epsilon$-transitions:]\quad\\
from a state $q_1$ of $\Ba$-state $p_{i_1}$
it decreases the counters $i_1$, $k_1$ and $l_2$ 
and moves to state $q_{\tau,\tau_1,\tau_2}^1$\\
from a state $q_2$ of $\Ba$-state $p_{i_2}$ 
it decreases the counters $i_2$, $k_2$ and $l_1$ 
and moves to state $q_{\tau,\tau_1,\tau_2}^2$\\
from a state $q_{\tau,\tau_1,\tau_2}$ 
it increases the counters $i$, $j_1$ and $j_2$ 
and moves to a state $q$ of $\Ba$-state $p_i$

\item[\rm up-transition:]
$(q_{\tau,\tau_1\,\tau_2}^1,a,q_{\tau,\tau_1,\tau_2}^2,q_{\tau,\tau_1,\tau_2})$.
\end{description}

\paragraph{\bf Correctness.}  
We argue as in the previous cases with the following modifications.

From $\dtree_1$ and $\dtree_2$ we first apply a bijection making sure that the
data values $d_1$ and $d_2$ of their roots are distinct and that $d_1$ has
$\Ba$-state $p_{\ell_1}$ in $\btree_2\otimes\dtree_2$ and $d_2$ has $\Ba$-state $p_{\ell_2}$ in
$\btree_1\otimes\dtree_1$. Moreover $\dtree_1$ and $\dtree_2$ share a common data value $d$
distinct from $d_1$ and $d_2$ of $\Ba$-state $p_{k_2}$ in $\dtree_2$ and
$\Ba$-state $p_{k_1}$ in $\dtree_1$.

For each $\theta \in \chi$ we select the associated data values making sure
that they are neither $d$, $d_1$ nor $d_2$. The decrement in the $\epsilon$-transitions make sure
that this is always possible. We then perform the same identification as in the
previous case. 
The rest of the argument is similar after setting $\dtree=d(\dtree_1,\dtree_2)$.

\begin{figure}
\small
\[
\begin{array}{cccc}
\tree'=\btree\otimes\dtree = 
\hspace{-8mm}
\begin{minipage}[t]{24mm}
\vspace{-3.5mm}
\begin{tikzpicture}
  [level distance=7mm,
   level 1/.style={sibling distance=12mm},
   level 2/.style={sibling distance=4mm}]
\node {$(b, d)$} 
 child { node {$(b_1, d_1)$} { node[itria] {} } }
 child { node {$(b_2, d_2)$} { node[itria] {} } };
\end{tikzpicture}
\end{minipage}
& 
\tree'[d] = 
\hspace{-8mm}
\begin{minipage}[t]{24mm}
\vspace{-3.5mm}
\begin{tikzpicture}
  [level distance=7mm,
   level 1/.style={sibling distance=12mm},
   level 2/.style={sibling distance=4mm}]
\node {$b (p_i)$} 
 child { node {$\# (p_{k_1})$} { node[itria] {} } }
 child { node {$\# (p_{k_2})$} { node[itria] {} } };
\end{tikzpicture}
\end{minipage}
& 
\tree'[d_1] = 
\hspace{-6mm}
\begin{minipage}[t]{24mm}
\vspace{-3.5mm}
\begin{tikzpicture}
  [level distance=7mm,
   level 1/.style={sibling distance=12mm},
   level 2/.style={sibling distance=4mm}]
\node {$\# (p_{j_1})$} 
 child { node {$b_1 (p_{i_1})$} { node[itria] {} } }
 child { node {$\# (p_{\ell_1})$} { node[itria] {} } };
\end{tikzpicture}
\end{minipage}
& 
\tree'[d_2] = 
\hspace{-6mm}
\begin{minipage}[t]{24mm}
\vspace{-3.5mm}
\begin{tikzpicture}
  [level distance=7mm,
   level 1/.style={sibling distance=12mm},
   level 2/.style={sibling distance=4mm}]
\node {$\# (p_{j_2})$} 
 child { node {$\# (p_{\ell_2})$} { node[itria] {} } }
 child { node {$b_2 (p_{i_2})$} { node[itria] {} } };
\end{tikzpicture}
\end{minipage}
\end{array}
\]
\caption{Proof of Theorem~\ref{prop-reduct}, Case 5.}
\label{fig-reduct5}
\end{figure}

\paragraph{\bf 6.}
  \emph{$\Ea$ guessed that the data value $d$ of the
  current node is different from the ones of its children but appear in both
  subtrees, with $p_{k_1}$ and $p_{k_2}$ as associated $\Ba$-states in respectively the left and right subtree. 
  Moreover $\Ea$ guessed that the data value 
  of the right child of the current node appear in its
  left subtree, with $p_{\ell_1}$ as associated $\Ba$-state in this left subtree, 
  and that the data value 
  of the left child does not appear in the right subtree. 
  Note that $p_{k_1}$, $p_{k_2}$ and $p_{\ell_1}$ must be \#-states in $Q_\Ba \setminus \{ p_\#\}$.}

\noindent
To handle this case for all transitions 
$\tau = (p_{k_1},b,p_{k_2}) \ra p_{i}$,
$\tau_1 = (p_{i_1},\#,p_{\#})\ra p_{j_1}$ and 
$\tau_2 = (p_{\ell_1},\#,p_{i_2})\ra p_{j_2}$ of $\Ba$, 
where none of $p_{i_1},p_{i_2}, p_i$ are $\#$-states but
$p_{k_1},p_{k_2},p_{\ell_1},p_{j_1},p_{j_2}$ are $\#$-states, $\Ea$ has the following transitions:
\begin{description}
\item[\rm $\epsilon$-transitions:]\quad\\
from a state $q_1$ of $\Ba$-state $p_{i_1}$
it decreases the counters $i_1$, $k_1$ and $l_1$ 
and moves to state $q_{\tau,\tau_1,\tau_2}^1$\\
from a state $q_2$ of $\Ba$-state $p_{i_2}$ 
it decreases the counters $i_2$, $k_2$ 
and moves to state $q_{\tau,\tau_1,\tau_2}^2$\\
from a state $q_{\tau,\tau_1,\tau_2}$
it increases the counters $i$, $j_1$ and $j_2$ 
and moves to a state $q$ of $\Ba$-state $p_i$

\item[\rm up-transition:]
$(q_{\tau,\tau_1\,\tau_2}^1,a,q_{\tau,\tau_1,\tau_2}^2,q_{\tau,\tau_1,\tau_2})$.
\end{description}

\paragraph{\bf Correctness.}  
We argue as in the previous cases with the following modifications.

From $\dtree_1$ and $\dtree_2$ we first apply a bijection making sure that the
data values $d_1$ and $d_2$ of their roots are distinct and that $d_1$ does not
appear in $\dtree_2$ and $d_2$ has $\Ba$-state $p_{\ell_1}$ in
$\btree_1\otimes\dtree_1$. Moreover $\dtree_1$ and $\dtree_2$ share a common data value $d$
distinct from $d_1$ and $d_2$ of $\Ba$-state $p_{k_2}$ in $\btree_2\otimes\dtree_2$ and
$\Ba$-state $p_{k_1}$ in $\btree_1\otimes\dtree_1$.

For each $\theta \in \chi$ we select the associated data values making sure
that they are neither $d$, $d_1$ nor $d_2$. The decrement in the $\epsilon$-transitions make sure
that this is always possible. We then perform the same identification as in the
previous case. 
The rest of the argument is similar after setting $\dtree=d(\dtree_1,\dtree_2)$.

\begin{figure}
\small
\[
\begin{array}{cccc}
\tree'=\btree\otimes\dtree = 
\hspace{-8mm}
\begin{minipage}[t]{24mm}
\vspace{-3.5mm}
\begin{tikzpicture}
  [level distance=7mm,
   level 1/.style={sibling distance=12mm},
   level 2/.style={sibling distance=4mm}]
\node {$(b, d)$} 
 child { node {$(b_1, d_1)$} { node[itria] {} } }
 child { node {$(b_2, d_2)$} { node[itria] {} } };
\end{tikzpicture}
\end{minipage}
& 
\tree'[d] = 
\hspace{-8mm}
\begin{minipage}[t]{24mm}
\vspace{-3.5mm}
\begin{tikzpicture}
  [level distance=7mm,
   level 1/.style={sibling distance=12mm},
   level 2/.style={sibling distance=4mm}]
\node {$b (p_i)$} 
 child { node {$\# (p_{k_1})$} { node[itria] {} } }
 child { node {$\# (p_{k_2})$} { node[itria] {} } };
\end{tikzpicture}
\end{minipage}
& 
\tree'[d_1] = 
\hspace{-6mm}
\begin{minipage}[t]{24mm}
\vspace{-3.5mm}
\begin{tikzpicture}
  [level distance=7mm,
   level 1/.style={sibling distance=12mm},
   level 2/.style={sibling distance=4mm}]
\node {$\# (p_{j_1})$} 
 child { node {$b_1 (p_{i_1})$} { node[itria] {} } }
 child { node {$\# (p_\#)$} { } };
\end{tikzpicture}
\end{minipage}
& 
\tree'[d_2] = 
\hspace{-6mm}
\begin{minipage}[t]{24mm}
\vspace{-3.5mm}
\begin{tikzpicture}
  [level distance=7mm,
   level 1/.style={sibling distance=12mm},
   level 2/.style={sibling distance=4mm}]
\node {$\# (p_{j_2})$} 
 child { node {$\# (p_{\ell_1})$} { node[itria] {} } }
 child { node {$b_2 (p_{i_2})$} { node[itria] {} } };
\end{tikzpicture}
\end{minipage}
\end{array}
\]
\caption{Proof of Theorem~\ref{prop-reduct}, Case 6.}
\label{fig-reduct6}
\end{figure}

\paragraph{\bf 7.}  
\emph{$\Ea$ guessed that the data value $d$ of the
  current node is different from the ones of its children but appears in both
  subtrees with $p_{k_1}$ and $p_{k_2}$ as associated $\Ba$-states. 
  Moreover it guessed that the data values of both children of the current node do not appear elsewhere.
  Note that $p_{k_1}$, $p_{k_2}$ must be \#-states in $Q_\Ba \setminus \{ p_\#\}$.}

\noindent
To handle this case for all transitions 
$\tau = (p_{k_1},b,p_{k_2}) \ra p_{i}$,
$\tau_1 = (p_{\#},b,p_{p_2})\ra p_{j_1}$ and 
$\tau_2 = (p_{i_1},b,p_{\#})\ra p_{j_2}$ of $\Ba$, 
where none of $p_{i_1},p_{i_2}, p_i$ are $\#$-states but
$p_{k_1},p_{k_2},p_{j_1},p_{j_2}$ are $\#$-states, $\Ea$ has the following
transitions:
\begin{description}
\item[\rm $\epsilon$-transitions:]\quad\\
from a state $q_1$ of $\Ba$-state $p_{i_1}$
it decreases the counters $i_1$, $k_1$ 
and moves to state $q_{\tau,\tau_1,\tau_2}^1$\\
from a state $q_2$ of $\Ba$-state $p_{i_2}$ 
it decreases the counters $i_2$, $k_2$ 
and moves to state $q_{\tau,\tau_1,\tau_2}^2$\\ 
from a state $q_{\tau,\tau_1,\tau_2}$ 
it increases the counters $i$, $j_1$ and $j_2$ 
and moves to a state $q$ of $\Ba$-state $p_i$

\item[\rm up-transition:]
$(q_{\tau,\tau_1\,\tau_2}^1,a,q_{\tau,\tau_1,tau_2}^2,q_{\tau,\tau_1,\tau_2})$.
\end{description}

\paragraph{\bf Correctness.}  
We argue as in the previous cases with the following modifications.

From $\dtree_1$ and $\dtree_2$ we first apply a bijection making sure that the
data values $d_1$ and $d_2$ of their roots are distinct and that $d_1$ does not
appear in $\dtree_2$ and $d_2$ does not appear in
$\dtree_1$. Moreover $\dtree_1$ and $\dtree_2$ share a common data value $d$
distinct from $d_1$ and $d_2$ of $\Ba$-state $p_{k_2}$ in $\dtree_2$ and
$\Ba$-state $p_{k_1}$ in $\dtree_1$.

For each $\theta \in \chi$ we select the associated data values making sure
that they are neither $d$, $d_1$ nor $d_2$. The decrement in the $\epsilon$-transitions make sure
that this is always possible. We then perform the same identification as in the
previous case. 
The rest of the argument is similar after setting $\dtree=d(\dtree_1,\dtree_2)$.

\begin{figure}
\small
\[
\begin{array}{cccc}
\tree'=\btree\otimes\dtree = 
\hspace{-8mm}
\begin{minipage}[t]{24mm}
\vspace{-3.5mm}
\begin{tikzpicture}
  [level distance=7mm,
   level 1/.style={sibling distance=12mm},
   level 2/.style={sibling distance=4mm}]
\node {$(b, d)$} 
 child { node {$(b_1, d_1)$} { node[itria] {} } }
 child { node {$(b_2, d_2)$} { node[itria] {} } };
\end{tikzpicture}
\end{minipage}
& 
\tree'[d] = 
\hspace{-8mm}
\begin{minipage}[t]{24mm}
\vspace{-3.5mm}
\begin{tikzpicture}
  [level distance=7mm,
   level 1/.style={sibling distance=12mm},
   level 2/.style={sibling distance=4mm}]
\node {$b (p_i)$} 
 child { node {$\# (p_{k_1})$} { node[itria] {} } }
 child { node {$\# (p_{k_2})$} { node[itria] {} } };
\end{tikzpicture}
\end{minipage}
& 
\tree'[d_1] = 
\hspace{-6mm}
\begin{minipage}[t]{24mm}
\vspace{-3.5mm}
\begin{tikzpicture}
  [level distance=7mm,
   level 1/.style={sibling distance=12mm},
   level 2/.style={sibling distance=4mm}]
\node {$\# (p_{j_1})$} 
 child { node {$b_1 (p_{i_1})$} { node[itria] {} } }
 child { node {$\# (p_\#)$} { } };
\end{tikzpicture}
\end{minipage}
& 
\tree'[d_2] = 
\hspace{-6mm}
\begin{minipage}[t]{24mm}
\vspace{-3.5mm}
\begin{tikzpicture}
  [level distance=7mm,
   level 1/.style={sibling distance=12mm},
   level 2/.style={sibling distance=4mm}]
\node {$\# (p_{j_2})$} 
 child { node {$\# (p_\#)$} { } }
 child { node {$b_2 (p_{i_2})$} { node[itria] {} } };
\end{tikzpicture}
\end{minipage}
\end{array}
\]
\caption{Proof of Theorem~\ref{prop-reduct}, Case 7.}
\label{fig-reduct7}
\end{figure}

\paragraph{\bf 8.} 
\emph{$\Ea$ guessed that the data value $d$ of the
  current node is different from the ones of its children and does not appear
  in the subtrees.  Moreover $\Ea$ guessed that the data values of both
  children of the current node are equal.}

\noindent
To handle this case for all transitions 
$\tau = (p_{\#},b,p_{\#}) \ra p_{i}$,
$\tau' = (p_{i_1},\#,p_{i_2})\ra p_{j}$ of $\Ba$, 
where none of $p_{i_1},p_{i_2},
p_i$ are $\#$-states, $\Ea$ has the following transitions:
\begin{description}
\item[\rm $\epsilon$-transitions:]\quad\\
from a state $q_1$ of $\Ba$-state $p_{i_1}$
it decreases the counters $i_1$ 
and moves to state $q_{\tau,\tau'}^1$\\
from a state $q_2$ of $\Ba$-state $p_{i_2}$ 
it decreases the counter $i_2$ 
and moves to state $q_{\tau,\tau'}^2$\\
from a state $q_{\tau,\tau'}$ 
it increases the counters $i$, and $j$ 
and moves to a state $q$ of $\Ba$-state $p_i$

\item[\rm up-transition:] $(q_{\tau,\tau'}^1,a,q_{\tau,\tau'}^2,q_{\tau,\tau'})$.
\end{description}

\paragraph{\bf Correctness.}  
We argue as in the previous cases with the following modifications.

From $\dtree_1$ and $\dtree_2$ we first apply a bijection making sure that the
data values of their roots are equal (let us call it $d_1$). 

For each $\theta \in \chi$ we select the associated data values making sure
it is not $d_1$. The decrement in the $\epsilon$-transitions make sure
that this is always possible. We then perform the same identification as in the
previous case. 
The rest of the argument is similar after setting
$\dtree=d(\dtree_1,\dtree_2)$, where $d$ is a fresh new value.

\begin{figure}
\small
\[
\begin{array}{ccc}
\tree'=\btree\otimes\dtree = 
\hspace{-8mm}
\begin{minipage}[t]{24mm}
\vspace{-3.5mm}
\begin{tikzpicture}
  [level distance=7mm,
   level 1/.style={sibling distance=12mm},
   level 2/.style={sibling distance=4mm}]
\node {$(b, d)$} 
 child { node {$(b_1, d_1)$} { node[itria] {} } }
 child { node {$(b_2, d_1)$} { node[itria] {} } };
\end{tikzpicture}
\end{minipage}
& 
\tree'[d] = 
\hspace{-8mm}
\begin{minipage}[t]{24mm}
\vspace{-3.5mm}
\begin{tikzpicture}
  [level distance=7mm,
   level 1/.style={sibling distance=12mm},
   level 2/.style={sibling distance=4mm}]
\node {$b (p_i)$} 
 child { node {$\# (p_\#)$} { } }
 child { node {$\# (p_\#)$} { } };
\end{tikzpicture}
\end{minipage}
& 
\tree'[d_1] = 
\hspace{-6mm}
\begin{minipage}[t]{24mm}
\vspace{-3.5mm}
\begin{tikzpicture}
  [level distance=7mm,
   level 1/.style={sibling distance=12mm},
   level 2/.style={sibling distance=4mm}]
\node {$\# (p_{j})$} 
 child { node {$b_1 (p_{i_1})$} { node[itria] {} } }
 child { node {$b_2 (p_{i_2})$} { node[itria] {} } };
\end{tikzpicture}
\end{minipage}
\end{array}
\]
\caption{Proof of Theorem~\ref{prop-reduct}, Case 8.}
\label{fig-reduct8}
\end{figure}

\paragraph{\bf 9.} 
\emph{$\Ea$ guessed that the data value $d$ of the
  current node is different from the ones of its children and does not appear in both subtrees. 
  Moreover $\Ea$ guessed that the data values of both children of the current node (say $d_1$ and $d_2$) 
  are distinct but appear in the other subtree with respective associated $\Ba$-state $p_{\ell_1}$ 
  and $p_{\ell_2}$.
  Note that $p_{\ell_1}$, $p_{\ell_2}$ must be $\#$-states in $Q_\Ba \setminus \{ p_\#\}$.}

\noindent 
This case is treated as before with the expected transitions.

\begin{figure}
\small
\[
\begin{array}{cccc}
\tree'=\btree\otimes\dtree = 
\hspace{-8mm}
\begin{minipage}[t]{24mm}
\vspace{-3.5mm}
\begin{tikzpicture}
  [level distance=7mm,
   level 1/.style={sibling distance=12mm},
   level 2/.style={sibling distance=4mm}]
\node {$(b, d)$} 
 child { node {$(b_1, d_1)$} { node[itria] {} } }
 child { node {$(b_2, d_2)$} { node[itria] {} } };
\end{tikzpicture}
\end{minipage}
& 
\tree'[d] = 
\hspace{-8mm}
\begin{minipage}[t]{24mm}
\vspace{-3.5mm}
\begin{tikzpicture}
  [level distance=7mm,
   level 1/.style={sibling distance=12mm},
   level 2/.style={sibling distance=4mm}]
\node {$b (p_i)$} 
 child { node {$\# (p_\#)$} { {} } }
 child { node {$\# (p_\#)$} { {} } };
\end{tikzpicture}
\end{minipage}
& 
\tree'[d_1] = 
\hspace{-6mm}
\begin{minipage}[t]{24mm}
\vspace{-3.5mm}
\begin{tikzpicture}
  [level distance=7mm,
   level 1/.style={sibling distance=12mm},
   level 2/.style={sibling distance=4mm}]
\node {$\# (p_{j_1})$} 
 child { node {$b_1 (p_{i_1})$} { node[itria] {} } }
 child { node {$\# (p_{\ell_1})$} { node[itria] {} } };
\end{tikzpicture}
\end{minipage}
& 
\tree'[d_2] = 
\hspace{-6mm}
\begin{minipage}[t]{24mm}
\vspace{-3.5mm}
\begin{tikzpicture}
  [level distance=7mm,
   level 1/.style={sibling distance=12mm},
   level 2/.style={sibling distance=4mm}]
\node {$\# (p_{j_2})$} 
 child { node {$\# (p_{\ell_2})$} { node[itria] {} } }
 child { node {$b_2 (p_{i_2})$} { node[itria] {} } };
\end{tikzpicture}
\end{minipage}
\end{array}
\]
\caption{Proof of Theorem~\ref{prop-reduct}, Case 9.}
\label{fig-reduct9}
\end{figure}

\paragraph{\bf 10.} 
\emph{$\Ea$ guessed that the data value $d$ of the
  current node is different from the ones of its children and does not appear in both subtrees. 
  Moreover $\Ea$ guessed that the data value $d_2$ of the right child 
  in its left subtree with $p_{\ell_1}$ as associated $\Ba$-state in this left subtree
  and that the data value $d_1$ of
  the left child does not appear in the right subtree. 
  Note that $p_{\ell_1}$ must be a $\#$-states in $Q_\Ba \setminus \{ p_\#\}$.}

\noindent 
This case is treated as before with the expected transitions.

\begin{figure}
\small
\[
\begin{array}{cccc}
\tree'=\btree\otimes\dtree = 
\hspace{-8mm}
\begin{minipage}[t]{24mm}
\vspace{-3.5mm}
\begin{tikzpicture}
  [level distance=7mm,
   level 1/.style={sibling distance=12mm},
   level 2/.style={sibling distance=4mm}]
\node {$(b, d)$} 
 child { node {$(b_1, d_1)$} { node[itria] {} } }
 child { node {$(b_2, d_2)$} { node[itria] {} } };
\end{tikzpicture}
\end{minipage}
& 
\tree'[d] = 
\hspace{-8mm}
\begin{minipage}[t]{24mm}
\vspace{-3.5mm}
\begin{tikzpicture}
  [level distance=7mm,
   level 1/.style={sibling distance=12mm},
   level 2/.style={sibling distance=4mm}]
\node {$b (p_i)$} 
 child { node {$\# (p_\#)$} { {} } }
 child { node {$\# (p_\#)$} { {} } };
\end{tikzpicture}
\end{minipage}
& 
\tree'[d_1] = 
\hspace{-6mm}
\begin{minipage}[t]{24mm}
\vspace{-3.5mm}
\begin{tikzpicture}
  [level distance=7mm,
   level 1/.style={sibling distance=12mm},
   level 2/.style={sibling distance=4mm}]
\node {$\# (p_{j_1})$} 
 child { node {$b_1 (p_{i_1})$} { node[itria] {} } }
 child { node {$\# (p_\#)$} { } };
\end{tikzpicture}
\end{minipage}
& 
\tree'[d_2] = 
\hspace{-6mm}
\begin{minipage}[t]{24mm}
\vspace{-3.5mm}
\begin{tikzpicture}
  [level distance=7mm,
   level 1/.style={sibling distance=12mm},
   level 2/.style={sibling distance=4mm}]
\node {$\# (p_{j_2})$} 
 child { node {$\# (p_{\ell_1})$} { node[itria] {} } }
 child { node {$b_2 (p_{i_2})$} { node[itria] {} } };
\end{tikzpicture}
\end{minipage}
\end{array}
\]
\caption{Proof of Theorem~\ref{prop-reduct}, Case 10.}
\label{fig-reduct10}
\end{figure}

\paragraph{\bf 11.} 
\emph{We omit the symmetric cases.}
\end{proof}

\section{From \texorpdfstring{\ebvass}{EBVASS} to \texorpdfstring{\fotwo}{FO2}}\label{sec-counter-fotwo}

We show in this section that reachability of \ebvass 
can be expressed as a sentence of \fotwo. 
This concludes the loop of reductions, showing that 
reachability for \ebvass, satisfiability of \fotwo and emptiness of \dad 
are equivalent as decision problems. 
The proof essentially mimics the reduction from \bvass to \fotwo described 
in~\cite{BojanczykMSS09jacm} with
extra material in order to handle the extra features.

\begin{theorem}\label{thm-bvass-fodeux}
 The reachability problem for \ebvass reduces to the satisfiability problem for \fotwo.
\end{theorem}

\begin{proof}
Given an \ebvass $\Ea=(Q,\A,q_0,k,\delta,\chi)$ and a state $q\in Q$, 
we compute a sentence $\phi \in \fotwo$ such that $\phi$ has a model 
iff the configuration $(q,v_0)$ is reachable in some tree
(where $v_0$ is the function setting all counters to $0$).

\noindent
We associate to $\Ea$ the following finite alphabet
$\A_\Ea = \delta \cup\{ D_i, I_i \mid 1\leq i \leq k\} 
               \cup \{ T_\theta, L_\theta, R_\theta \mid \theta \in \chi \}$. 
Intuitively $D_i$ says that the counter
$i$ has been decreased, $I_i$ says that the counter $i$ has been increased,
$\delta$ encodes the transition relation, and the letters $L_\theta$, $R_\theta$ and $T_\theta$ will be used to enforce 
the constraint $\theta$. 
The formula $\phi$ we construct accept all binary data trees 
of $\Trees(\A_\Ea\times \D)$ encoding runs of $\Ea$. It turns out that $\phi$
accepts more trees but any accepted trees of $\phi$ can be transformed into an accepting
run of $\Ea$ with simple transformations.

\noindent
We start with the encoding of a single transition $\mu \in \delta$.

\noindent
If $\mu$ is an $\epsilon$-transition then we encode it with two nodes
$x,y$ where $y$ is the unique child of $x$ and the label of $x$ is $\mu$ while
the label of $y$ is $D_i$ (resp. $I_i$) 
if $\mu$ was decreasing (resp. increasing) counter $i$.

\noindent
If $\mu$ is an up-transition, then we encode
it as a subtree of the following form:
\begin{itemize}
\item The root has label $\mu$,
\item below the root there is a (vertical) sequence of nodes of arity one whose 
  labels form a word of
  $\displaystyle\sum_{\theta = \merc{i_1}{i_2}{i}\in\chi} (I_iT_\theta)^*$, 
  where $\sum$ denotes concatenation,
\item the last node of that sequence has arity two and two branches starts from that node,
\item the sequence of labels of the left branch forms a word of
$\displaystyle\sum_{\theta=\merc{i_1}{i_2}{i}\in\chi} (D_{i_1}L_\theta D_{i_2}R_\theta)^*$,
\item the sequence of labels of the right branch forms a word of
$\displaystyle\sum_{\theta=\merc{i_1}{i_2}{i}\in\chi} (D_{i_1}R_\theta D_{i_2}L_\theta)^*$,
\item for all $\theta$ the number of occurrences of $L_\theta, R_\theta$ and
  $T_\theta$ are the same.
\end{itemize}
A tree satisfying all these items, except maybe the last one, is said to be a \emph{pseudo-encoding} of the up-transition $\mu$. 
Notice that pseudo-encodings of up-transitions form a regular tree language.

\noindent
From there, the encoding of a run is obtained in the obvious way by concatenating
encodings of transitions.

The formula $\phi$ essentially describes this construction. It first enforces
that the tree has the desired shape: 
\begin{itemize}
\item The tree is a repetition of a sequence of the form: a pseudo-encoding 
of one up-transition followed by the encodings of several $\epsilon$-transitions,
\item the sequencing is valid: if $\mu$ and $\nu$ are consecutive transitions
  in the tree then the starting state of one is the ending state of the other,
\item the initial state $q_0$ can be found at the leaves 
and the state $q$ is reached at the root.
\end{itemize}
Note that the above three conditions can be checked by a standard tree
automaton over $\A_\Ea$, and therefore can be expressed in $\emso^2(\descendant,\neighbor)$.
Therefore, by setting $\A=\A_c\times\A'$ for a
suitable $\A'$ matching the existential part of the $\emso$ formula, 
the property above can be expressed in $\fo^2(\descendant,\neighbor)$.

The formula $\phi$ now needs to make sure that no counter ever gets negative and that
pseudo-encodings of up-transitions are actually real encodings. 
This is where
data values are needed:
The formula $\phi$ enforces that 
\begin{enumerate}[label=\enspace(\arabic*)]
\item \label{it:Di} 
      no two nodes with label $D_i$
      can have the same data value, for $1 \leq i \leq k$,
\item \label{it:Ii} 
      no two nodes with label $I_i$
      can have the same data value, for $1 \leq i \leq k$,
\item \label{it:DiIi} 
      for all $i\in[k]$, every node with label $D_i$ has a descendant 
      with label $I_i$ and with the same data value,
\item \label{it:IiDi} 
      for all $i\in[k]$, every node with label $I_i$ has an ancestor 
      with label $D_i$ and with the same data value.
\end{enumerate}
These four conditions enforce that the counters never get negative and that
they are all set to $0$ at the root.  It remains to enforce that all
pseudo-encodings can be transformed into real encodings.  
This is done with the following conditions.
\enlargethispage{\baselineskip}

\begin{enumerate}[label=\enspace(\arabic*)]
\setcounter{enumi}{4}
\item \label{it:Ttheta} 
      no two nodes with label $T_\theta$, for $\theta \in \chi$,
      can have the same data value,
\item \label{it:Ltheta} 
       no two nodes with label $L_\theta$, for $\theta \in \chi$, 
       can have the same data value,
\item \label{it:Rtheta} 
       no two nodes with label $R_\theta$, for $\theta \in \chi$, 
       can have the same data value,
\item \label{it:TthetaLRtheta} 
      every node with label $T_\theta$ has a descendant with label $L_\theta$ 
      and a descendant with label $R_\theta$
      both with the same data value,
\item \label{it:LRthetaTtheta} 
      every node with label $L_\theta$ or $R_\theta$ has an ancestor with label $T_\theta$ 
      and with the same data value,
\item \label{it:LRtheta} 
  two nodes of label $L_\theta$ and $R_\theta$ with the same data value
  are not comparable with the ancestor relationship.
\end{enumerate}
It now remains to show that $\phi$ has the desired property.
\begin{lemma}\label{lemma-ebvass-fo}
$\phi$ has a model iff $(q, v_0)$ is reachable by $\Ea$.
\end{lemma}
\begin{proof}
{\bf From reachability to models of $\phi$.}
Assume that $(q, v_0)$ is reachable and let $\rho$ be a run of $\Ea$ witnessing this fact.
Let $\atree$ be the tree constructed from $\rho$ by concatenating the sequences
of encodings of transitions of $\rho$ as explained above. 
The binary tree $\atree$ certainly satisfies the ``regular'' part of $\phi$. 
We now assign the data values so that
the remaining part of $\phi$ is satisfied. This is done in the obvious way:
each time a counter $i$ is decremented, as the resulting value is positive,
this means that a matching increment was performed before. Similarly, each time
a constraint $\theta$ is used in a transition $\mu$, we assign one distinct
data value per triple $L_\theta,R_\theta,T_\theta$ 
occurring in the encoding of
$\mu$. The formula was constructed to make the resulting tree a model of $\phi$.

{\bf From models of $\phi$ to reachability.}
Assume now that $\tree=\atree\otimes \dtree \models \phi$. 
Unfortunately, it may happen that $\atree$ 
does not encode a run of $\Ea$ because 
some section corresponds to a pseudo-encoding of an up-transition, 
instead of an expected real encoding. 
However, we show that from $\tree$ we can construct another tree
$\tree'=\atree'\otimes\dtree'$ such that $\tree' \models \phi$ 
and $\atree'$ encodes a real run of $\Ea$.

To see this, let us consider a node $x$ of $\tree$ 
with label $T_\theta$, 
where $\theta = \merc{i_1}{i_2}{i}$, and let $d = \dtree(x)$. 
Let $x_1$ and $x_2$ be two descendants of $x$ 
with respective labels $L_\theta$ and $R_\theta$
and such that $d = \dtree(x_1) = \dtree(x_2)$. 
Let $z$ be the least common ancestor of $x_1$ and $x_2$.
The existence of $x_1$ and $x_2$ is guaranteed by $\phi$ 
(conditions \ref{it:Ttheta}--\ref{it:TthetaLRtheta}). 
The sentence $\phi$ also ensures
that $x$ is an ancestor of $z$ 
(conditions \ref{it:LRthetaTtheta}--\ref{it:LRtheta}).  
By construction the subtree at $z$ must correspond to
a pseudo-encoding of an up-transition $\mu'$.

We now move (down) $x$ and its parent 
(that must have label $I_i$) 
right above $z$ within the coding of $\mu'$.
Similarly we move (up) $x_1$ and its parent
(that must have label $D_{i_1}$) 
right below $z$, and similarly for $x_2$.  
The reader can verify that the
resulting tree is still a model of $\phi$: the regular conditions remain 
obviously satisfied.  
Conditions \ref{it:Di}--\ref{it:IiDi}  
are still valid because the node of label $I_{i_1}$ 
matching the parent of $y$ was already below the initial position of
$y$ and its new position is upward in the tree. 
Finally conditions \ref{it:Ttheta}--\ref{it:LRtheta} 
remain valid by construction.

Repeating this argument eventually yields a model $\tree'=\btree\otimes\dtree'$
of $\phi$ such that $\btree$ is a correct sequencing of encodings of
transitions a $\Ea$. 
This encoding is actually a real run because 
conditions \ref{it:Di}--\ref{it:IiDi} 
of $\phi$ immediately enforces that no counter is ever negative.
\end{proof}
Theorem~\ref{thm-bvass-fodeux} is now immediate from Lemma~\ref{lemma-ebvass-fo}.
\end{proof}

\section{Conclusion}\label{sec-conclusion}
We have seen that satisfiability of \fotwo, emptiness of \dad and reachability
of \ebvass are equivalent problems in terms of decidability. The main open
problem is of course whether they are all decidable or not.

The use of the \ebvass constraints of the form 
$\merc{i_1}{i_2}{i}$
is crucial for the construction of Section~\ref{sec-dad-counter}.
Their semantics cannot be directly simulated with the usual BVASS,
but it is not clear whether \ebvass are strictly more expressive than BVASS,
and whether this extension is needed in order 
to capture the expressive power of \fotwo on data trees. 

In our definition of \ebvass the constraints of the form $\merc{i_1}{i_2}{i}$
have a ``commutative'' semantics. 
Without commutativity, \textit{i.e.}, the rule modifies
only counter $i_1$ on the left child and counter $i_2$ on the right child,
the automata model is more powerful. 
In order to describes its runs as in the
proof of Theorem~\ref{thm-bvass-fodeux}, the logic needs to be able to enforce
that a $L_\theta$ must be to the left of the $R_\theta$ with the same data value. 
This can be done by adding the document order predicate into the logic.
A close inspection of the proof of Theorem~\ref{th-fotwodad} and
Theorem~\ref{prop-reduct} then shows that the extension of \fotwo with the
document order predicate can be captured by a \dad without the commutativity
rule and that such automata can be captured by the non-commutative version of
\ebvass.

In~\cite{BjSch10} it was shown that, over data words, the Data Automata model
of~\cite{BDMSS11} is more expressive than the Register Automata
of~\cite{KF94}. It is not obvious that our automata model \dad extends the
expressive power of the straightforward extension of register automata to data
trees. This remains to be investigated.

\bibliographystyle{abbrv} 

\bibliography{FO2BVASS} 

\end{document}